%% file: main.tex
\documentclass[letterpaper,twocolumn,10pt]{article}
\usepackage{usenix2019_v3}

\usepackage{algorithm}
\usepackage{algorithmicx}
\usepackage[noend]{algpseudocode}
\usepackage{amsmath}
\usepackage{amssymb}
\usepackage{amsthm}
\usepackage[utf8]{inputenc}
\usepackage{color}
\usepackage{graphicx}
\usepackage[subtle]{savetrees}
\usepackage{subfigure}
\usepackage{xfrac}
\usepackage{xspace}

\newcommand{\name}{RFocus\xspace}
\newtheorem{theorem}{Theorem}
\newtheorem{lemma}{Lemma}

\title{\Large \bf \name{}: Practical Beamforming for Small Devices} 
\date{}
\author{
{\rm Venkat Arun and Hari Balakrishnan}\\
MIT, CSAIL
}

\newcommand{\va}[1]{}
\newcommand{\hb}[1]{}
\newcommand{\outline}[1]{}

\begin{document}

\maketitle

\input{abstract}
\input{intro}
\input{related}
\input{overview}
\input{design}
\input{algo}

\input{phys}
\input{eval}

\input{concl}
\input{acks}

\bibliographystyle{abbrv}
\bibliography{cite}

\end{document}

%% file: abstract.tex
\begin{abstract}

To reduce transmit power, increase throughput, and improve communication range, radio systems---such as IoT sensor networks, Wi-Fi and cellular networks---benefit from the ability to direct their signals, to ensure that more of the transmitted power reaches the receiver. Many modern systems beamform with antenna arrays for this purpose. However, a radio's ability to direct its signal is fundamentally limited by its size. Unfortunately practical challenges limit the size of modern radios, and consequently, their ability to beamform. In many settings, radios on devices must be small and inexpensive; today, these settings are unable to benefit from high-precision beamforming. 

To address this problem, we introduce {\em RFocus}, which moves beamforming functions from the radio endpoints to the environment. RFocus includes a two-dimensional surface with a rectangular array of simple elements, each of which functions as an RF switch. Each element either lets the signal through or reflects it. The surface does not emit any power of its own. The state of the elements is set by a software controller to maximize the signal strength at a receiver, with a novel optimization algorithm that uses signal strength measurements from the receiver. The RFocus surface can be manufactured as an inexpensive thin wallpaper, requiring no wiring. This solution requires only a method to communicate received signal strengths periodically to the RFocus controller. Our prototype implementation improves the median signal strength by $10.5\times$, and the median channel capacity by $2.1\times$. 

\end{abstract}

%% file: intro.tex
\section{Introduction}


Many radio systems use directional or sectorized antennas and beamforming to improve the throughput or range of a wireless communication link. Beamforming ensures that a larger fraction of transmitted energy reaches the intended receiver, while reducing unintended interference. It is well known that a radio with many antennas spread densely over a large area can fundamentally beamform better than a smaller radio~(\S\ref{s:size-tradeoff}).

However, there are many practical challenges to making radio systems with large antenna arrays. First devices such as IoT sensors and handhelds must be small in size. Second, connecting each antenna in an array to full-fledged radio transmit/receive circuitry increases cost and power. Third, large, bulky systems are hard to deploy, even in infrastructure base stations or access points. 

To address these challenges and achieve the equivalent of a large number of antennas, we propose {\em \name{}}. \name is a {\em programmable mirror/lens} placed in the environment that configures itself to direct a signal from a transmitter to a receiver. This approach moves the task of beamforming from the transmitter to the environment. Any device in the vicinity can reap the benefits of \name{}'s size, without itself being large or consuming additional energy. 

\name{} is made up of thousands of simple elements organized in a rectangular array. To minimize cost and power-consumption, each element only contains a single 2-way RF switch, which is as inexpensive as a passive RFID tag. It can be manufactured as a thin, flexible sheet that can be pasted on walls as (painted) wallpaper. It can be manufactured to be battery and wire-free, powered and controlled with RF signals.

Each element in the \name{} surface can be in one of two states. When ``on'', the signal is reflected; otherwise, the signal passes through. \name{} doesn't emit any power of its own. Each receiver periodically sends measurements to a \name{} controller. The controller uses these to configure the \name{} surface to maximize signal strength between the pairs. To do so, the controller has a training phase, where it  configures test states on the surface, and monitors the changes in the reported measurements. Once it has enough data, it uses an optimization algorithm to set a state that maximizes the endpoints' objectives.


This controller's optimization algorithm solves three key challenges. First, indoor environments exhibit complex multi-path. Therefore the optimal configuration might not correspond to directing the signal along a single direction. Second, one way to measure the effect of each element on the channel is to vary it individually, and measure the difference. However, the effect of an individual element is miniscule, making it very hard to measure (\S\ref{s:eval:micro:measure}). Third, drift in carrier frequency offset (CFO drift) corrupts an already weak signal, making it hard to measure phase. \va{Do we mention that commodity devices usually don't report phase?}


Our prototype has 3,720 inexpensive\footnote{At scale, each of the antenna elements is on the order of a few cents or less.} antennas on a 6 square-meter surface. We believe this configuration may be the largest number of antennas ever used to improve communication links.  

\name{} can work both as a mirror or a lens, with the controller seamlessly choosing the right mode. That is, radio endpoints can be on the same side of the surface, or on opposite sides. Further, the optimization organically prioritizes more important elements, which makes \name{} robust to element failure.

\va{Mobility and fast evolution is out of scope}

\va{Good place to mention that \name{} is unlikely to hurt performance? We can do an eval on this. Manikanta had this concern.}

The contributions of this paper include:
\begin{enumerate}

\item The \name{} controller, incorporating three key ideas. First, to handle complex multipath, \name{} exploits the fact that the \name{} surface design is approximately linear. Second, it modulates all elements at once to boost the effect of the \name{} surface on the channel, hence making the change large enough to be measurable. Third, it relies only on signal strength measurements, sidestepping difficulties in measuring phase. Under realistic assumptions, we prove that the controller finds a solution that is within $2\times$ of the optimal.

\item The \name{} surface, which has several desirable properties. First, unlike prior work,~\cite{laia-nsdi,laia-hotnets} to make our reflectors inexpensive, we use just two states. We prove that this achieves at-least $1 / \pi$ of the performance of reflectors that have infinitely many states (\S\ref{s:phys:num-states}). Second, our reflectors' sizes are comparable to a wavelength, which considerably simplifies their design. There are diminishing returns to making them much smaller than a wavelength (\S\ref{s:phys:pixel-size}). Third, our design is both area-efficient and approximately linear~(\S\ref{s:eval:micro}). To achieve this, we exploit the well-known phenomenon that waves ignore details that are much smaller than a wavelength~(\S\ref{s:phys:design}). Finally, it works across a wide range of frequencies.

\item Experiments, which show that in a typical indoor office environment, \name{} achieves a median $10.5\times$ improvement in signal strength and $2\times$ improvement in channel capacity.

\end{enumerate}

\hb{the third contribution above is fine to start but the three points aren't too clear. Moreover, you really need some experimental result}

%% file: related.tex
\section{Related Work}

LAIA~\cite{laia-nsdi,laia-hotnets} is a recent project that helps endpoints whose line-of-sight path is blocked by a wall. LAIA deploys elements connected by a wire going through holes in the wall, taking energy collected from one side of the wall to the other. Thus individual LAIA elements have a much larger impact on the channel than \name{} reflectors. This allows LAIA to function with fewer elements than \name{}. However, LAIA's benefits are limited to receivers which are blocked from the sender by the wall which its elements traverse. By contrast to LAIA, our goal is to design a system that is not just limited to such receivers. To achieve this goal, we solve the challenge of measuring the effect of the \name{} reflectors. Additionally, our mirrors use a 2-way RF-switch~\cite{rf-switch} which is $20\times$ cheaper than the phase-shifters used by LAIA~\cite{laia-phase-shifter}. We prove that \name{} achieves a channel improvement that is at-least $1 / \pi$ as that of a surface whose elements have infinite states~(\S\ref{s:phys:num-states}). 

The idea of using passive controllable reflectors in the environment to improve communication links has been explored before~\cite{ee-1-with-exp,ee-2,ee-3-theory,hypersurface,hypersurface-2,hypersurface-3}. This prior research is, however, is in theory or simulation, except for one paper providing a preliminary experimental result~\cite{ee-1-with-exp}. To the best of our knowledge, \name{} is the first large-scale real-world prototype of such a system. 

One line of theoretical work~\cite{hypersurface,hypersurface-2,hypersurface-3} explores using metamaterial elements to create antenna arrays, where each ``pixel'' is much smaller than a wavelength. This approach offers fine-grained control over the electric field at the surface. This work designs sophisticated algorithms to solve Maxwell's equations to reason about the surface. We argue, however, that at distances greater than a few wavelengths from the surface, such fine-grained control gives only incremental benefits~\S\ref{s:phys:pixel-size}. Hence we adopt a simpler design with larger, non-metamaterial pixels.

Another related line of work improves wireless coverage by analyzing the indoor space, and custom-designing a 3D reflector for that space. When 3D-printed and placed behind the access point (AP), the reflector reflects energy in specific directions to maximize signal strength at previously uncovered areas\cite{3d-print-reflect-1,3d-print-reflect-2}. Once deployed, this reflector has a very low operational cost; it is just a passive metal-coated object. But a new reflector needs to be designed for each new location and whenever the space changes or the AP is moved. Further, one solution has to be designed for all pairs of endpoints, whereas the \name{} design can dynamically design a new pattern for each pair of endpoints at runtime.

Range extenders are an alternate technique to increase signal strength at the receiver. However, by retransmitting each packet, they increase interference and waste transmission opportunities. By precisely focusing energy already available, \name{} decreases interference while increasing signal strength.

Reconfigurable antennas~\cite{reconf-ant-review} and reflectarray antennas~\cite{refarr-ant-review} have RF switches and phase shifters, which allow them to dynamically change their characteristics such as operating frequency, input impedance and directionality. These approaches focus on modifying an antenna to get better characteristics. By contrast, we leave the transmit and receive antennas unmodified, instead modifying the environment to improve communication for all nearby devices. Thus one set of antennas in the environment (\name surface) can serve multiple devices, even though the devices are too small to have a large antenna array.

%% file: overview.tex
\section{Overview}
\begin{figure}
    \centering
    \includegraphics[width=\linewidth]{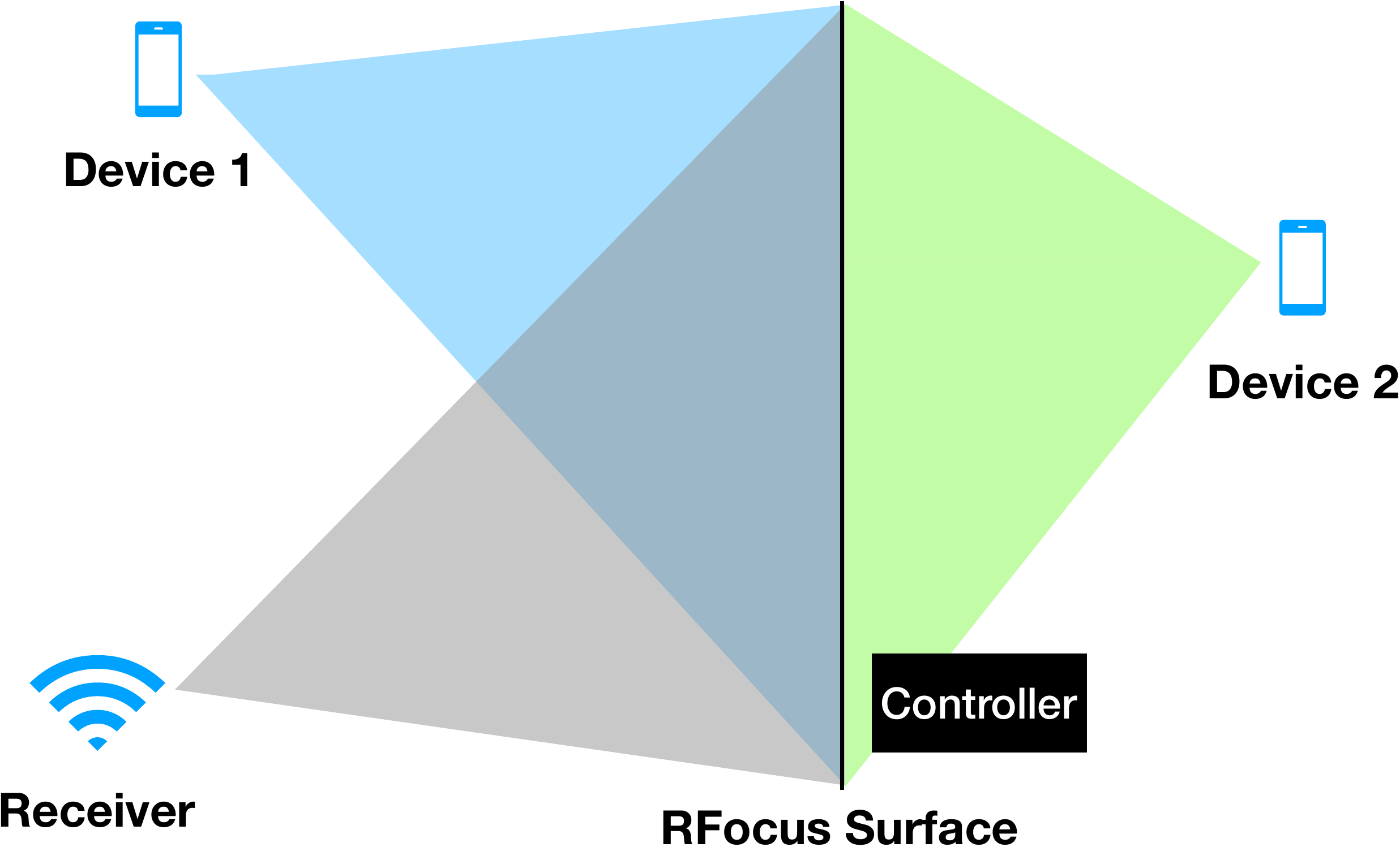}
    \caption{\name{} consists of a large passive reflecting surface and a controller that configures how the surface reflects signals by setting the ``on'' or ``off'' state of each RF switch on the surface. Endpoints use \name{} for beamforming transmissions by periodically sending received signal strength measurements to the controller. The controller uses this information in an optimization algorithm that determines the state of each switch on the surface to focus the transmitter's signal at the intended receiver. \name{} does not need to know where the endpoints are, and can improve SNR when the endpoints are on the same side of the surface (like a mirror) or on opposite sides (like a lens). Once configured, it can switch between different configurations in $\approx 1$ ms, allowing different pairs of endpoints to time-share the surface's beamforming abilities.}
    \label{fig:overview}
\end{figure}
\begin{figure}
\centering
\includegraphics[width=\linewidth]{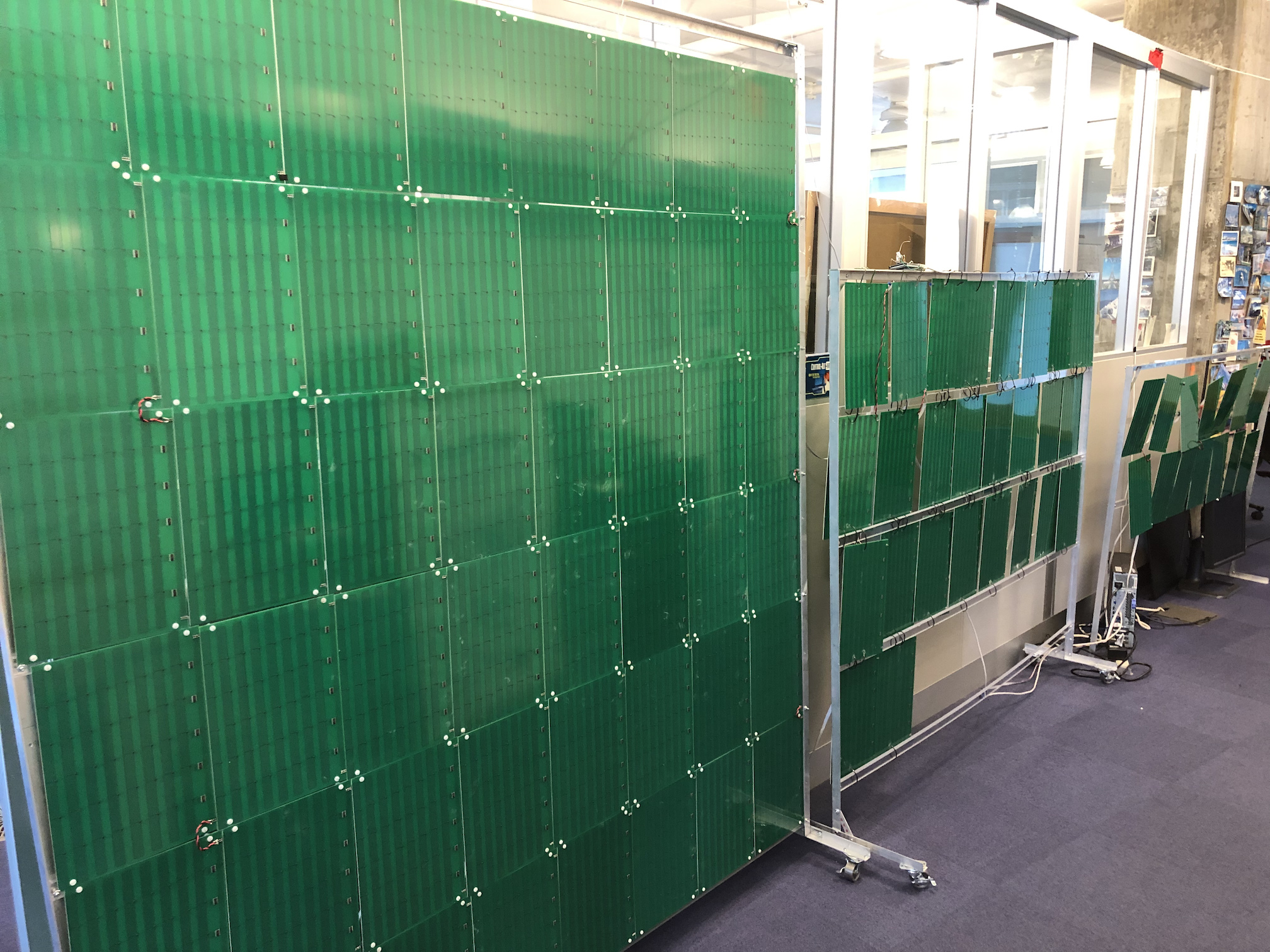}
\caption{Our prototype of the \name surface.}
\label{fig:photo}
\end{figure}

\name{} seeks to increase the received signal strength of a transmitter's signal. A transmitter can take advantage of this capability by reducing its transmit power, which helps reduce interference between nearby devices, and helps battery-operated devices conserve energy. Communicating endpoints may also use the boost in signal strength to achieve higher throughput or range.

Figure~\ref{fig:overview} shows how \name{} enables multiple communicating endpoints to exploit the \name{} surface's large area (a large area is essential for precisely directing radio signals, as explained in~\S\ref{s:size-tradeoff}. Figure~\ref{fig:photo} is a picture of our prototype. 

To model the system, we treat each element in the \name{} surface as contributing a phase to the received signal~(\S\ref{s:design:model}). \name{}'s behavior toward radio signals is controlled by which elements are turned on, as this controls where the paths interfere constructively versus destructively. For instance, we can programmatically control the angle at which the signal is reflected or the point at which it is focused. The controller uses signal-strength measurements from the receiver to determine which elements to turn on, so that all the paths interfere constructively at the receiver. To do so, it sets random configurations on the surface, and monitors how the signal strength reported by the endpoints varies, and uses an efficient majority-voting algorithm to converge to a configuration that is within $2\times$ of optimal. 

%% file: design.tex
\section{Background}

\subsection{System Model and Notation}
\label{s:design:model}
In any environment, there are multiple paths between any two antennas. For a narrow-band signal of wavelength $\lambda$, the effect of each path can be represented by a complex number. It is a function of the path length and any reflectors it encounters. The net effect of the channel is the sum of the effects of all the paths. A subset of the paths pass via each of the $N$ elements on the RFocus surface. Denote the channel effect of the elements by $c_1, \ldots, c_N$. We combine all the paths {\it not} going via \name{} into one complex number $c_E$ ($E$ for environment). The net channel is therefore:
\begin{equation}
\label{eqn:phy-model}
h = c_E + \sum_{i=1}^N \alpha_i c_i
\end{equation}
Here $\alpha_i$ represents the amplitude change and phase shift introduced by element $i$. \name{} controls the channel, $h$, by adjusting $\alpha_i$. $c_E$ and $c_i$ are functions of the path lengths. $\alpha_i$ is a function of the state of the $i^{th}$ element and its neighbors, the shape of the antennas composing the element, and the angles at which the path enters and leaves the $i^{th}$ element. Since the elements are passive and do not contain an energy source, $|\alpha_i| \le 1$. If we had full control over each $\alpha_i$, then to maximize the channel strength, $|h|$, we would set $\alpha_i = \frac{c_E}{|c_E|}\frac{c_i^*}{|c_i|}$, where $c_i^*$ denotes conjugation. We get:

\begin{equation}
\label{eqn:phy-model-opt}
|h| = \left|c_E + \frac{c_E}{|c_E|}\left(\sum_{i=1}^N|c_i|\right)\right| = |c_E| + \sum_i |c_i|
\end{equation}

This assignment maximizes $|h|$. However, we do not have full control over $\alpha_i$. Our elements can only be in one of two states, on or off. If we assume that $\alpha_i$ is a function of only the $i^{th}$ element's state and not that of its neighbors, then we can write the channel as:
\begin{equation}
\label{eqn:base-model}
h = h_Z + \sum_{i=1}^N b_i h_i
\end{equation}
Section~\ref{s:eval:micro:linearity} shows that this is a good approximation.

Here $b_i \in \{0, 1\}$ denotes whether the $i^{th}$ element is on or off. $h_Z$ is the channel when all elements are off, and $h_i$ is the effect of turning the $i^{th}$ element on. Here, we have folded the complexities of $\alpha_i$ into $c_i$ to get $h_i$.  We prove that having the ability to set any $b_i \le 1$ gives only a $\pi \approx 3.14$ factor advantage in optimizing $|h|$ (or a $\pi^2$ factor advantage in $|h|^2$, the energy) over being restricted to pick $b_i \in \{0, 1\}$ (\S\ref{s:phys:num-states}).

\subsection{How Size Helps Communication}
\label{s:size-tradeoff}

\newcommand{\appropto}{\mathrel{\vcenter{
  \offinterlineskip\halign{\hfil$##$\cr
    \propto\cr\noalign{\kern2pt}\sim\cr\noalign{\kern-2pt}}}}}

In this section, we discuss some well-known results that illustrate why and how much the size of the antenna array helps.

\paragraph{Quadratic growth.} From equation~\ref{eqn:phy-model} we can see that, when optimized, $|h|$ grows linearly with the number of elements, $N$, if the $|h_i|$ values don't fall as $N$ increases. Hence, the {\it energy} in the signal, $|h|^2$, grows as $O(N^2)$. The whole is thus greater than the sum of its parts! This property holds because the phases of all paths are aligned, and energy that would have otherwise gone elsewhere is now focused on the receiver. If the phases are random, $|h|$ only grows as $O(\sqrt{N})$ (due to the central limit theorem), and the energy hence grows as $O(N)$, which matches our intuition about everyday objects which do not align phase. We show that even with our system, where we can't fully align all phases, and our elements are limited to having 2-states, the $O(N^2)$ growth still holds (\S\ref{s:phys:num-states}).

\begin{figure}
\includegraphics[width=\linewidth]{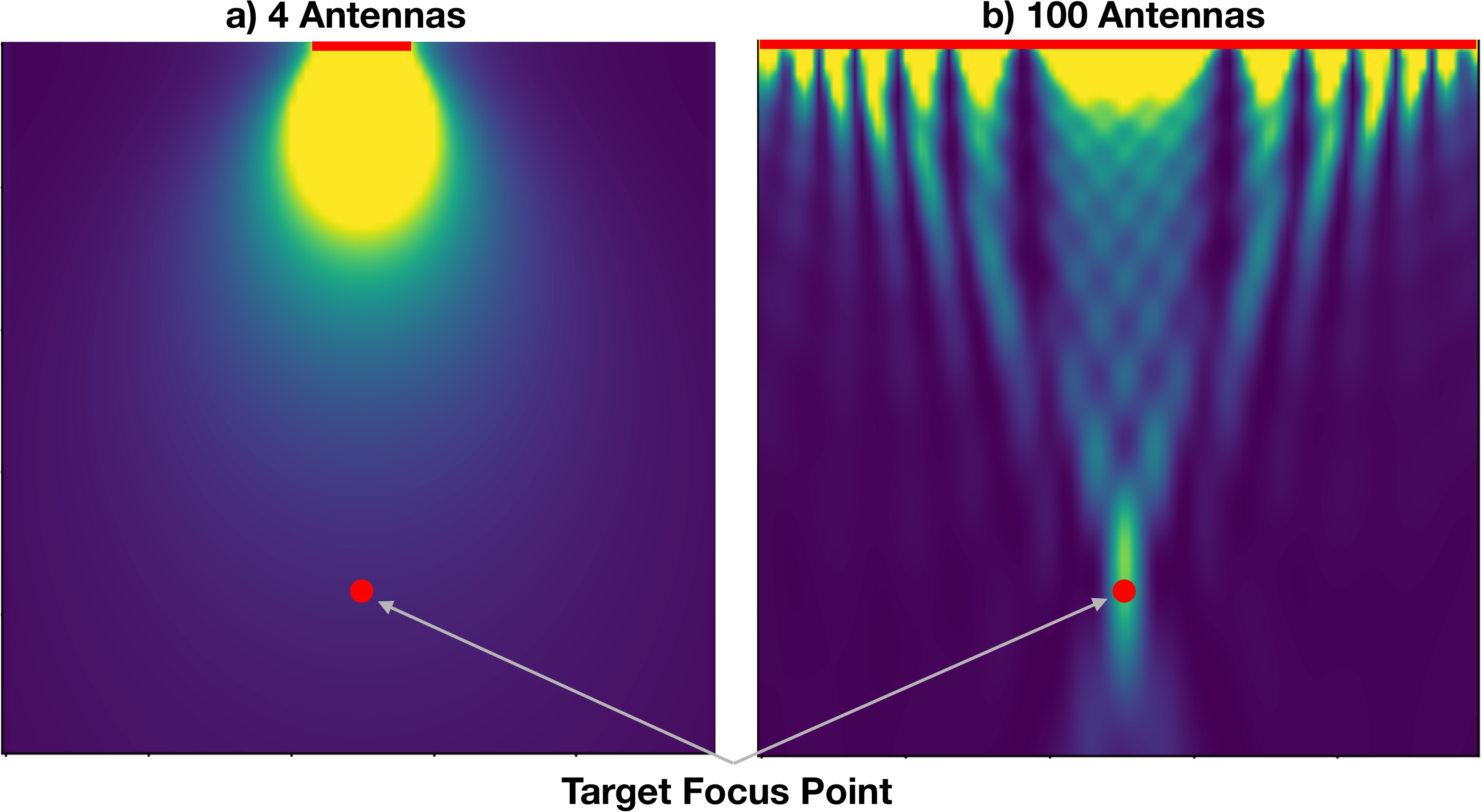}
\caption{A simulation of how signal strength is distributed when an antenna array tries to maximize signal strength at the target point. The antennas are in the highlighted regions. With 4 antennas, the signal is directional, but quickly begins to spread out. The 100 antennas, subtend a large angle at the target, and are hence able to focus energy there, avoiding attenuation due to spreading. For clarity, the very high signal strength near the antennas is not shown.}
\label{fig:diffraction-limit}
\end{figure}

\paragraph{The diffraction limit.} We now use well known results to quantitatively understand how the size of the \name{} surface allows it to focus on a small area. The Abbe diffraction limit states that, if our surface has an area $A$, and focuses energy at a point at distance $d$ away from it, the spot onto which the energy is focused will have an area $a = k \lambda^2(1 + 4\frac{d^2}{A})$, where $\lambda$ is the wavelength.~\va{Remember to explain in appendix} $k$ is a proportionality constant, conventionally set to $0.5$; it is a function of how well we want the energy to be confined within $a$. The energy is spread through the area $a$, hence the energy available to the receiver is $O(\Omega / a)$, where $\Omega$ is the solid angle subtended by \name{} on the transmitter, and accounts for the amount of transmitted energy incident on the surface. Hence \name{} works best when the transmitter (or receiver, since the channel is symmetric) is close (say at distance $d'$) to the surface, and therefore $\Omega$ is large. Now $\Omega \propto \arctan^2{(\sqrt{A}/d')}$, which when $\Omega$ isn't too large (say $< \pi / 2$), is $O(A)$.

Note, $A / d^2$ is proportional to the angle subtended by the surface on the receiver. We get two regimes depending on how large this angle is. If the radio is far away from the surface (i.e. $d^2 \gg A$), then $a \appropto d^2 / A$ and energy falls as $O(A\Omega / d^2)$. This is still a $1 / d^2$ fall, but the constant is improved by a factor of $A\Omega$. If the radio is closer, then $d^2 \sim A$, and the first term dominates. In this regime, $a$ can be made quite small, on the order of a few $\lambda^2$. Hence almost all of the transmitted energy can be incident on the receiver.

In traditional beamforming, $A$ is typically small, hence $d^2 \gg A$ and we are always in the first regime where the signal experiences a $1 / d^2$ attenuation as it spreads out. The difference between the two regimes is illustrated in Figure~\ref{fig:diffraction-limit}.

%% file: algo.tex
\section{Optimization Algorithm}
\label{s:algo}

The \name{} controller uses measurements from the radio endpoints to maximize signal strength at the receiver. In this section we will first describe the challenges in measuring changes in the channel, and why we rely solely on RSSI measurements~(\S\ref{s:algo:measurement}). We then describe how to compute the optimal solution given perfect information about the channel, and an easy way to get a 2-approximate solution to the optimal~(\S\ref{s:algo:perfect-info}). Finally we describe our algorithm, and prove that it finds a 2-approximate solution~(\S\ref{s:algo:rssi-algo}).

\subsection{Measuring the Channel}
\label{s:algo:measurement}
\paragraph{A direct, but naive, method: } When all the elements are turned off, the channel is $h_Z$, by definition. Ideally, to measure each $h_i$, we could turn just the $i^{th}$ element on, and measure the difference from $h_Z$. But this change is usually too feeble to measure, because an element is just a small piece of metal lying somewhere in the environment. When we tried measuring this effect, we had to block the direct path between the radios with a sheet of metal, and keep the element in the closest indirect path. Even then, we had to average over hundreds of measurements. In general, this small change is well beyond our ability to measure (\S\ref{s:eval:micro:measure}).

Prior work~\cite{laia-nsdi} is able to measure the effect of individual elements because the direct path was blocked by a wall. Further, each of their elements physically traversed the wall with wires connecting antennas on either side of the wall. This allows each of their elements to have large effect on the channel. 

\paragraph{Boosting the signal.} Each $h_i$ may be small, but all the elements together can have a large effect. We could generate several random configurations of the surface by randomly choosing the state of each element. If we vary $N$ elements, the variance of measurements among these states will have an expected magnitude of $\sqrt{N}\sigma$ (due to the central limit theorem) where $\sigma$ is the variance of each $h_i$. This gives us a $O(\sqrt{N})$ boost in amplitude, which is an $O(N)$ boost in energy. 

\paragraph{Challenges in measuring phase.} A second skeleton algorithm would measure the channel for many random configurations of the surface, and solve the linear equations to obtain all the $h_i$. However this is also hard, since it needs to measure changes in the {\em phase} of the channel. First, many commodity devices do not report phase~\va{is this ok, since we don't demonstrate commodity devices?}. Second, linear regression computes unnecessary information, and is hence sample inefficient; we just want to know whether to turn an element on or off; we don't care what the exact phase and amplitude of $h_i$ is. Third, and most importantly, the change in the channel, even after this $O(\sqrt{N})$ boost, is still quite small~\va{Explain why the array will give a good and measurable difference though the boost is small?}. The carrier frequency offset (CFO) drifts fast enough that it overpowers the measurement unless we compute the difference in the channel immediately before and after the surface's state changes. For this, we need to know exactly when the surface changed its configuration. This means the clock of the receiver must be synchronized with that of the controller. In our micro-benchmarks, we synchronized the clocks to within $30 \mu$s by connecting a wire from the Arduino controller to the receiver. It is only with such synchronization that we get reasonable phase-measurements. While acceptable for doing research, this is impractical for deployment. 

\paragraph{Using RSSI.} One way to avoid CFO-drift is to rely purely on signal-strength measurements, ignoring phase information. This has the additional benefit of being deployable with commodity receivers that report the RSSI without phase.
RSSI is not always an absolute metric, and may vary due to automatic-gain control changes or temperature changes on the amplifier. Hence we always measure RSSI of a test state {\em relative} to a baseline state; e.g., the all-zeros state where all elements are turned off. We call this is the {\em RSSI-ratio}.

We develop a simple algorithm that uses RSSI measurements to find a 2-approximation to the optimal. That is, the $|h|$ for the solution it finds is at least $|h_{OPT}| / 2$. But first, let us see how the controller would work if it knew $h_E$ and all the $h_i$.

\subsection{Optimal and Approximate Solutions Given Perfect Information}
\label{s:algo:perfect-info}
Assume that equation~\ref{eqn:base-model} is accurate. Later we verify it experimentally~(\S\ref{s:eval:micro:linearity}). The controller needs to assign values 0 or 1 to each $b_i$ such that $|h| = |h_E + \sum_i b_ih_i|$ is maximized. Let $h_{OPT}$ be an optimal solution. In this solution, $b_i = 1$ if and only if $h_i \cdot h_{OPT}^* \ge 0$ ($x^*$ denotes complex conjugation). Otherwise, we could flip $b_i$ to get $|h_{OPT} - b_ih_i + (1 - b_i)h_i| = |h_{OPT} + (1 - 2b_i)h_i| \ge |h_{OPT}|$. 

\begin{figure}
\centering
\includegraphics[width=0.55\linewidth]{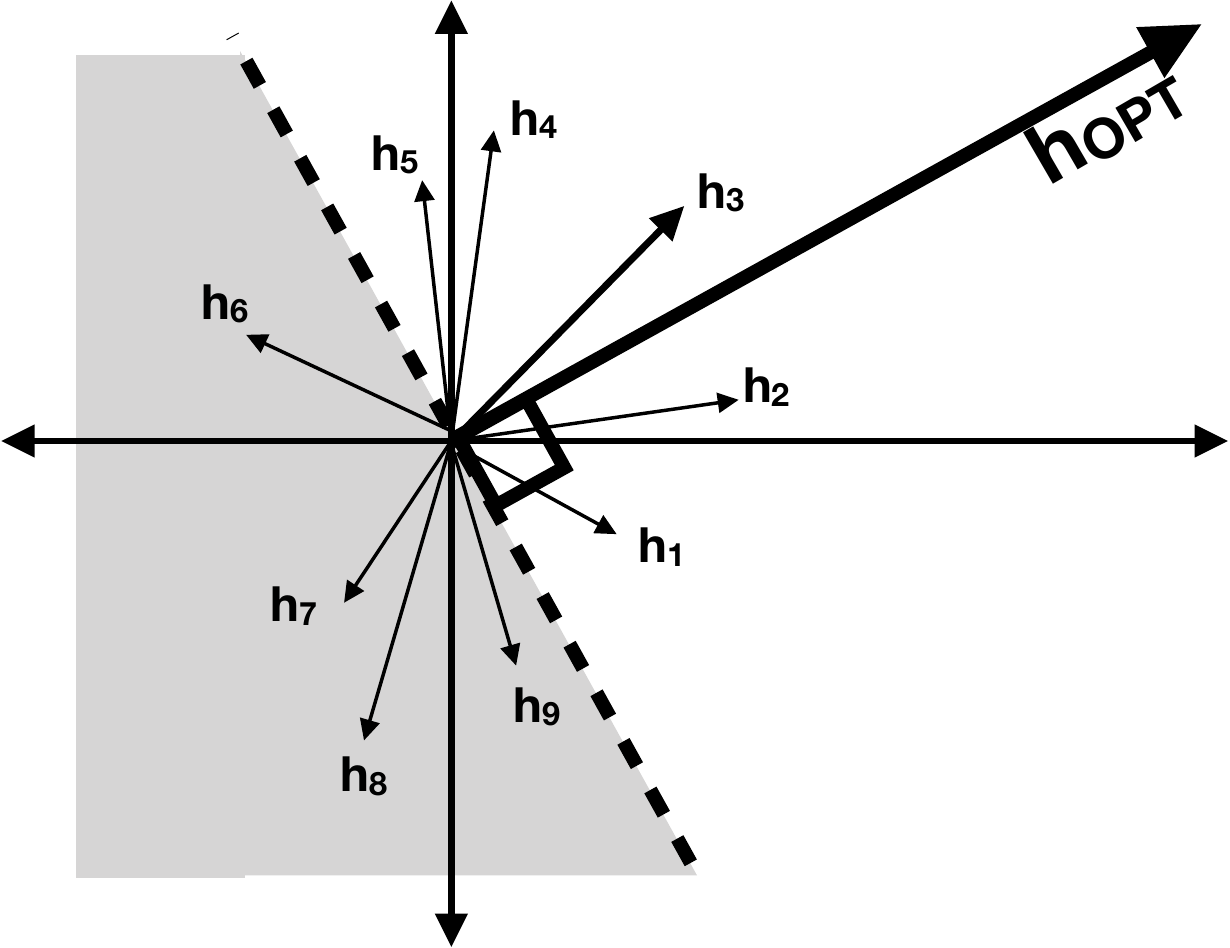}
\includegraphics[width=0.41\linewidth]{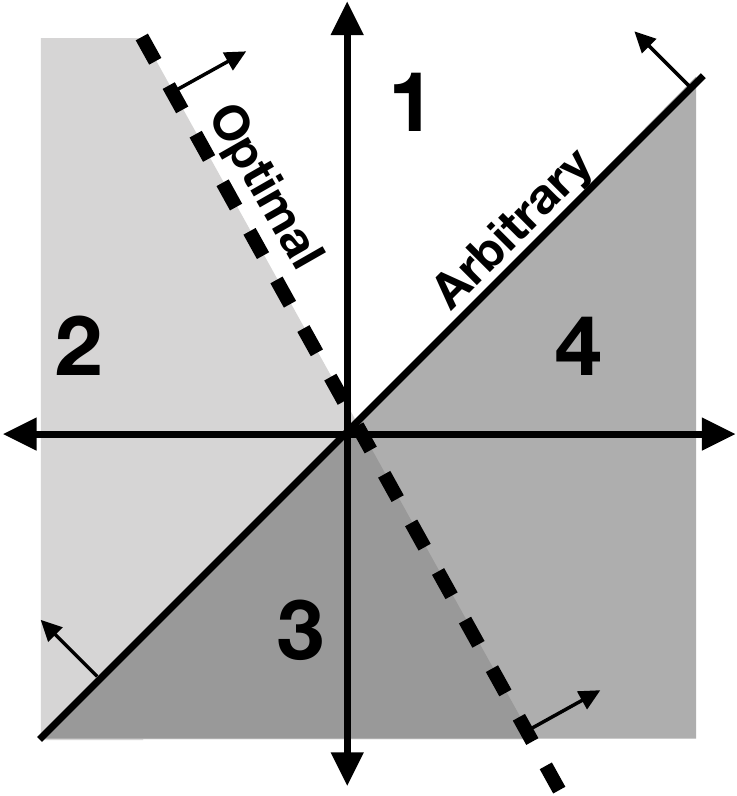}
\caption{The figures show the channel variables on a complex plane. On the left, the optimal solution consists of choosing all $h_i$ on on side of a halfplane (the non-greyed region), and setting the others to zero. This removes  destructive interference from the elements in the greyed-out region and improves signal strength. The precise choice of the halfplane is not very important, and an arbitrary choice (right) will give us a 2-approximation.}
\label{fig:halfplane}
\end{figure}

Therefore the optimal solution consists of the subset of $h_i$ that make an acute angle with $h_{OPT}$. Therefore all $h_i$ lie in one halfplane separated by a line passing through the origin as shown on the left in Figure~\ref{fig:halfplane}.
To find an optimal assignment, we need to guess the direction of $h_{OPT}$, and pick each $\alpha_i$ based on whether $h_i \cdot h_{OPT} \ge 0$. How do we guess the direction/phase of $h_{OPT}$? If we know $h_E$ and $h_i \, \forall i$, then this is easy. Simply iterate through all the directions\footnote{There are an infinite number of directions. But the resulting configuration changes only when the new line includes/excludes a new $h_i$. Hence we only need to consider $O(n)$ directions, and the entire algorithm is $O(n^2)$.} 

But we do not have all the $h_i$. The following lemma shows that the choice of the direction of $h_{OPT}$ is not very important. It will be useful when we develop an algorithm that works purely based on RSSI measurements.

\begin{lemma}
\label{lemma:simple-2-approx}
Consider an arbitrary line passing through the origin of the complex plane. It divides the elements in two sets, depending on which side their $h_i$ lies on. This gives two candidate solutions, where $b_i$ for elements in one set is set to 1 and the elements in the other is set to 0. Of these two, the solution that gives a higher value of $|h|$ is a 2-approximation. That is, $|h| \ge |h_{OPT}| / 2$. 
\end{lemma}
\begin{proof}
Consider the right side of Figure~\ref{fig:halfplane}. It shows the halfplane corresponding to the optimal solution, and a different, arbitrary, division of the halfplane. These divide the complex into 4 parts. Let the sum of $h_i$ of elements in these parts be $H_1, H_2, H_3$ and $H_4$. Then $h_{OPT} = H_1 + H_4$. The triangle inequality implies that $|H_1| + |H_4| \ge |H_1 + H_4| = |h_{OPT}|$. Hence either $|H_1| \ge |h_{OPT}| / 2$ or $|H_4| \ge |h_{OPT}|/2$. The sum of all the $h_i$ in the two sides of our arbitrary line is $H_1 + H_2$ and $H_3 + H_4$. Adding $H_2$ and $H_3$ to $H_1$ and $H_4$ respectively cannot decrease their magnitude, since the parts subtend an acute angle with each other. Hence at least one of the two sides of the arbitrary line has a sum with magnitude $\ge |h_{OPT}| / 2$.
\end{proof}

\subsection{The RSSI-Based Optimization Algorithm}
\label{s:algo:rssi-algo}

\newcommand{\RSSI}{\mathrm{{\tt RSSI}}}
\newcommand{\States}{\mathrm{{\tt States}}}
\newcommand{\VoteOn}{\mathrm{\tt VoteOn}}
\newcommand{\VoteOff}{\mathrm{\tt VoteOff}}
\newcommand{\Opt}{\mathrm{\tt Opt}}

\begin{algorithm}[H]
\caption{\name{}'s majority voting algorithm}
\label{alg:majority-vote}
\begin{algorithmic}
\Procedure{MajorityVote}{$\States, \ \  \RSSI$}
\State // $\RSSI[j]$ gives the measured RSSI-ratio for $\States[j]$
\State $\Opt \gets$ blank-list \Comment{The final optimized state}
\State $a \gets$ \Call{Average}{$\RSSI$}

\For {$i := 0$ {\bf to} $\mathrm{\tt NumElements}$}
\State $\VoteOn \gets 0, \ \  \VoteOff \gets 0$

\For {$n := 0$ {\bf to} $\States\mathrm{{\tt.len}}$}
\If{
\State ($\States[n][i] == 1$ {\bf and} $\RSSI[n] > a$) {\bf or} 
\State ($\States[n][i] == 0$ {\bf and} $\RSSI[n] < a$)}
\State $\VoteOn \gets \VoteOn + 1$
\Else
\State $\VoteOff \gets \VoteOff + 1$
\EndIf
\EndFor
\State $\Opt[i] \gets (\VoteOn > \VoteOff)$

\EndFor

\State // One of $\Opt$ and $\mathbf{not}\:\:\:\Opt$ is a 2-approximation
\State \Return ($\Opt,\:\: \mathbf{not}\:\:\:\Opt)$

\EndProcedure
\end{algorithmic}
\end{algorithm}

Given the measured RSSI-ratio for a set of $K$ random configurations, our majority-voting algorithm (see Algorithm~\ref{alg:majority-vote}) finds the state for each bit $i$. It compares the RSSI-ratio of each measurement to the average value: if it is higher (or lower) than the average RSSI-ratio when the $i^{th}$ element is on (or off), then it votes for the element to be on. Else it votes for the element to be off. The $i^{th}$ element's optimized state is determined by which value received more votes. As the following theorem shows, this algorithm effectively finds the 2-approximate solution discussed above.


\begin{theorem}

Assume the model given in equation~\ref{eqn:base-model} is correct, $N \rightarrow \infty, K \rightarrow \infty$ and $|h_i| \ll \left|h_Z + \frac{1}{2}\sum_{j=1}^Nh_j\right| \,\, \forall i$. $N$ is the number of elements and $K$ is the number of measurements. At-least one of the solutions returned by Algorithm~\ref{alg:majority-vote} produces a channel magnitude $|h|$ that is at least $\sfrac{1}{2}$ that of the optimal solution.
\end{theorem}
\begin{proof}
\begin{figure}
\includegraphics[width=\linewidth]{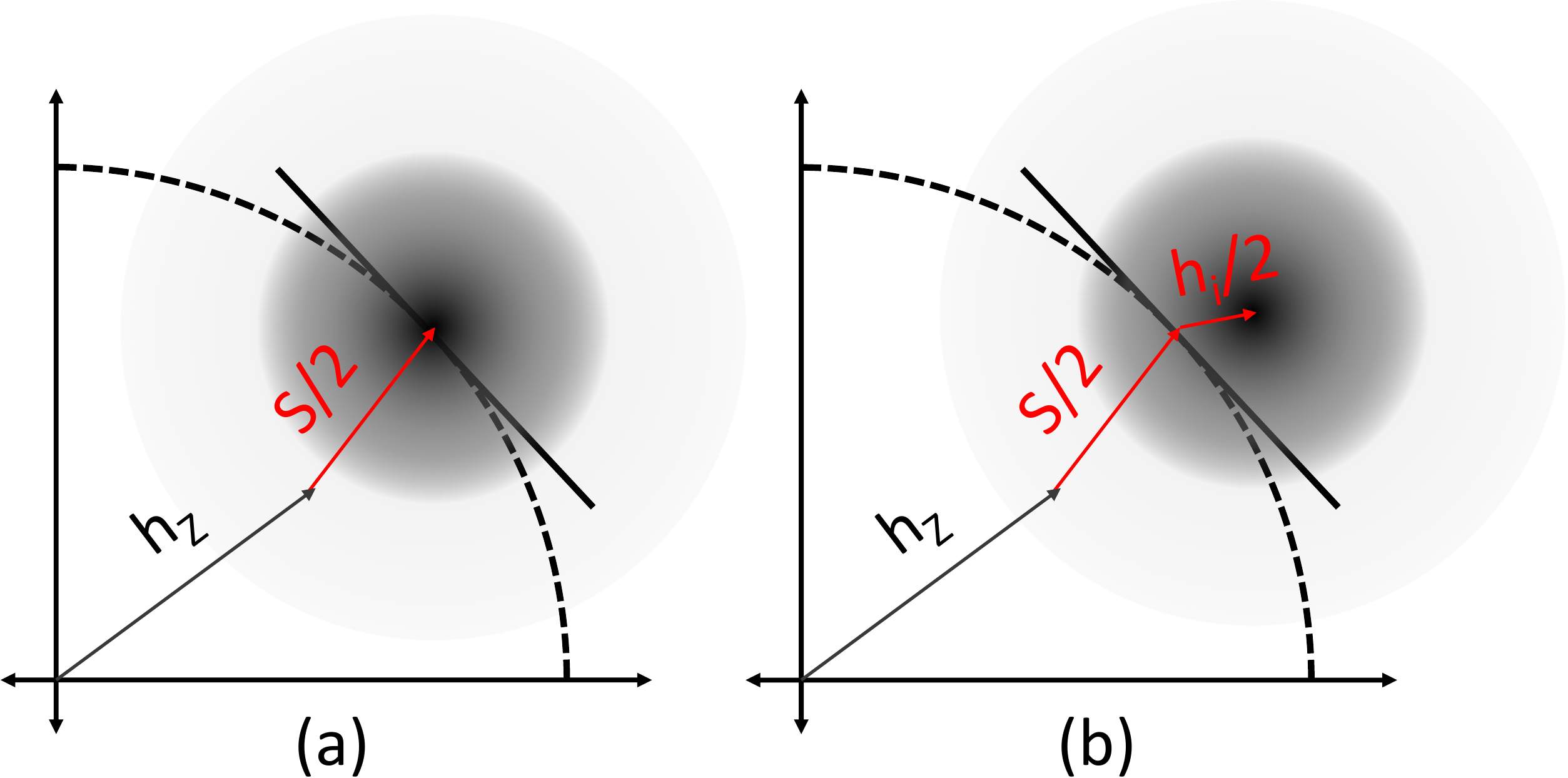}
\caption{(a) shows the probability density function of the channel when all bits are chosen uniformly at random. (b) shows the PDF conditioned on the $i^{th}$ bit being `on'. The dashed circle is centered at the origin with radius $|h_Z +S/2|$, where $S = \sum_{i=1}^N$ . Depending on which side of the circle $h_i$ takes the mean, the mean magnitude will be greater or lesser than that of the unconditional PDF.}
\label{fig:algo-mag}
\end{figure}

According to equation~\ref{eqn:base-model}, when we randomly vary $b_i$, $h$ becomes a random variable, $H$, with mean $h_Z + S/2$, where $S = \sum_{i=1}^Nh_i$ is the sum of the $h_i$ of the rest of the elements. The $S/2$ term appears because we include each element with probability half. Figure~\ref{fig:algo-mag}a) shows this probability distribution. Consider an element $i$ that hasn't yet been fixed. If we condition the probability distribution on the $i^{th}$ element being on, then the PDF shifts by $h_i/2$ as shown in Figure~\ref{fig:algo-mag}b), because the $i^{th}$ element's value is fixed in these samples (it shifts by $h_i/2$ and not $h_i$ since we have already included the other $h_i/2$ in $S/2$). If we condition on the element being off, then the mean shifts by $-h_i/2$.

The central limit theorem implies that $H / \sqrt{N}$ is Gaussian as $N \rightarrow \infty$. Hence, if $h_i/2$ (or $-h_i/2$) shifts the mean to be within the circle, then when the element is on (or off) the RSSI-ratio is more likely to be less than the unconditional mean. If it $h_i/2$ shifts the mean outside the circle, then the opposite holds. Hence, as $K \rightarrow \infty$, we can determine with confidence tending to 100\%, whether the conditional mean is inside or outside the circle, by looking at RSSI-measurements alone.

We assumed that $|h_i| \ll \left|h_Z + \frac{1}{2}\sum_{j=1}^Nh_j\right| \,\, \forall i$. This is reasonable since $h_i$ is just the effect of just one element, which is $O(\sqrt{N})$ times smaller compared to $S / 2$ as $N\rightarrow\infty$, if $h_i$ are i.i.d random variables. Given this assumption, to a good approximation, the mean shifts outside the circle if and only if it shifts to the outer side of the tangent line shown in figure~\ref{fig:algo-mag}. 

Thus, we can determine on which side of the tangent line an element lies, purely by looking  at the RSSI ratio. This allows us to get a 2-approximate solution as shown in Lemma~\ref{lemma:simple-2-approx}. 
\end{proof}

\paragraph{The Optimization Algorithm in Practice} In practice, we make a few changes to Algorithm~\ref{alg:majority-vote}: We cannot perform $K \rightarrow \infty$ measurements. So we stop varying an element as soon as its value is known with $> 95\%$ confidence, as determined by a two-sided Student's t-test. This enables us to get the benefits of an element, as soon as we are confident about its value. This gives us an organic method to prioritize higher-impact elements, since we will be confident of their value earlier. It also provides robustness against damaged/occluded elements, since they do not increase the time required to determine the other elements' values. We divide our evolution into batches of 2000 measurements. After each batch, we fix the values of elements we are sure about, and vary only the remaining ones\footnote{The unconditional mean changes, as we fix more and more elements. We account for this by comparing against the mean/median for RSSI-ratio of only that batch.}. Finally, we use median instead of mean, since it is more robust to one-shot measurement noise such as those due to packet loss and Automatic Gain Control (AGC) changes.

%% file: phys.tex
\section{Antenna Array Design}
In our design, we have made two key design choices. First, each reflector has only two states: one that reflects the signal, another that lets it through. Second, each element is half a wavelength tall and $1/10$ of a wavelength wide. In this section we explore the tradeoffs in these choices.

\subsection{How Big Should Each Element Be?}
\label{s:phys:pixel-size}

An array with many small antennas gives better control over the reflected signal, while one with fewer but larger antennas is cheaper and simpler. It is well understood that, in an array, the inter-antenna spacing should be smaller than half a wavelength, since otherwise there will be grating lobes, where the signal is sent in directions other than the desired direction. But is there a benefit to making the spacing even smaller? While our antennas are inexpensive enough that this may be cost effective, designing antennas that are much smaller than a wavelength~\footnote{Called ``electrically small antennas''} is challenging. Small antennas are either inefficient, absorbing only a small fraction of incident energy, or they are efficient only over a small bandwidth. Further, when placed close to each other, antennas interact strongly with each other in a way that is often hard to model.

Fortunately, a well understood result suggests that making antennas much smaller than half a wavelength will give only marginal benefits. Consider two infinite parallel planes a distance $d$ apart, separated by a homogeneous medium. Variations in electric/magnetic fields in one plane that are faster than once per wavelength, will have a negligible effect on the fields on the other plane (the effect they have decays exponentially with $d$). A formal statement and proof can be found in~\cite{angular-spectrum-representation}. Hence any fine-grained variations we introduce in the surface will be lost as soon as the signal propagates a few wavelengths in either direction. Hence we can design an array with antennas comparable to a wavelength, and still get most of the benefits.

\def \PixellationTheorem {
Consider the energy from a transmitter that is reflected onto the receiver from the surface. Assume both are far away from the sheet (the error falls exponentially with distance). We bound the ratio, $\epsilon$, between the energy reflected from an ideal surface (with infinitely many, infinitesimally small pixels) and that from a surface whose pixels' largest dimension is $a$:

\begin{equation}
\epsilon \ge \begin{cases}
	\left( \frac{1}{\sqrt{2}\pi}\frac{\sin{\pi \nu a}}{a\nu} \right) & a < \lambda = c / \nu \\
    0 & \mathrm{otherwise} \\
\end{cases}
\end{equation}
}

\subsection{How Many States Should Each Element Have?}
\label{s:phys:num-states}

In our design we chose elements that can be in only one of two states. But we could have chosen a design that offers greater control. Ideally we would be able to control the exact phase and amplitude with which each element reflects its energy. In terms of our model in equation~\ref{eqn:base-model}, we would have been able to set any $b_i \in \mathbb{C}, \, |b_i| \le 1$, instead of being restricted to $b_i \in \{0, 1\}$ ($|b_i| \le 1$ because \name{} doesn't emit any energy of its own). Let us denote the amount of energy that can be directed by the surface in the two cases as $h_{IDEAL}$ and $h_{REAL}$. In this section, we show that $|h_{REAL}| \ge |h_{IDEAL}| / \pi$. Hence, by having just two states, we can get $1/\pi^2$ of the signal strength we'll get with infinitely many states. Here we only consider the signal from the surface and not the direct path, $h_Z$. To maximize signal strength, we'd need to align the phases with $h_Z$ also, which is usually easy since the phases of $h_i$ are uniformly distributed, since the antennas are spaced closer than $\lambda/2$.

\begin{theorem}
$$|h_{REAL}| \ge \frac{|h_{IDEAL}|}{\pi}.$$
\end{theorem}
\begin{proof}

In the ideal case, we should align the phases of all $h_i$ to maximize $|h|$. Then, we get $|h_{IDEAL}| = \sum_{i=1}^N|h_i|$~(\S\ref{}).  Define $A = \int_{-\pi}^{\pi} \sum_{i=1}^N |h_i \cdot e^{j\theta}|\mathrm{d}\theta$. Intuitively, it computes the sum of components of $h_i$ along angle $\theta$ and integrates over all $\theta$. Each $\theta$ corresponds to a perpendicular to a halfplane, as discussed before in \S\ref{s:algo:perfect-info}. At least one of these, say $\theta_0$, corresponds to the optimal half-plane, wherein the optimal solution contains all the $h_i$ such that $h_i \cdot e^{j\theta_0} > 0$. These contribute $h_i \cdot e^{j\theta}$ toward $h_{REAL}$. Thus, $|h_{REAL}| = \sum_{i=1}^N \mathrm{Max}(0, h_i \cdot e^{j\theta_0})$, hence $|h_{REAL}| \ge \frac{1}{2} \sum_{i=1}^N|h_i \cdot e^{j\theta_0}|$, because otherwise we could have chosen $-\theta_0$ and obtained a better $|h_{REAL}|$. Hence $|h_{REAL}| \ge \frac{1}{2} \mathrm{max}_{\theta \in [\pi, \pi)} \sum_{i=1}^N |h_i \cdot e^{j\theta}| \ge \frac{1}{2} \frac{1}{2\pi} \int_{-\pi}^{\pi} \sum_{i=1}^N |h_i \cdot e^{j\theta_0}|~\mathrm{d}\theta = \frac{A}{4\pi}$.

We can rearrange the sum as $A = \sum_{i=1}^N \int_{-\pi}^\pi |h_i e^{j\theta}|~\mathrm{d}\theta = \sum_{i=1}^N |h_i| \left(\int_{-\pi}^\pi |\cos{\theta}|~ \mathrm{d}\theta\right)$. The second step is possible, because $\cos{\theta}$ is taking a dot-product of $h_i$ over a unit complex number with phase $\theta$. Since we are integrating over all angles, it doesn't matter which angle we start from. Now we can evaluate the integral to get $A = 4\sum_{i=1}^N|h_i|= 4|h_{IDEAL}|$. Since $|h_{REAL}| \ge \frac{A}{4\pi}$, $|h_{REAL}| \ge \frac{|h_{IDEAL}|}{\pi}$. 
\end{proof}


\subsection{Our Design}
\label{s:phys:design}

\begin{figure}
\includegraphics[width=\linewidth]{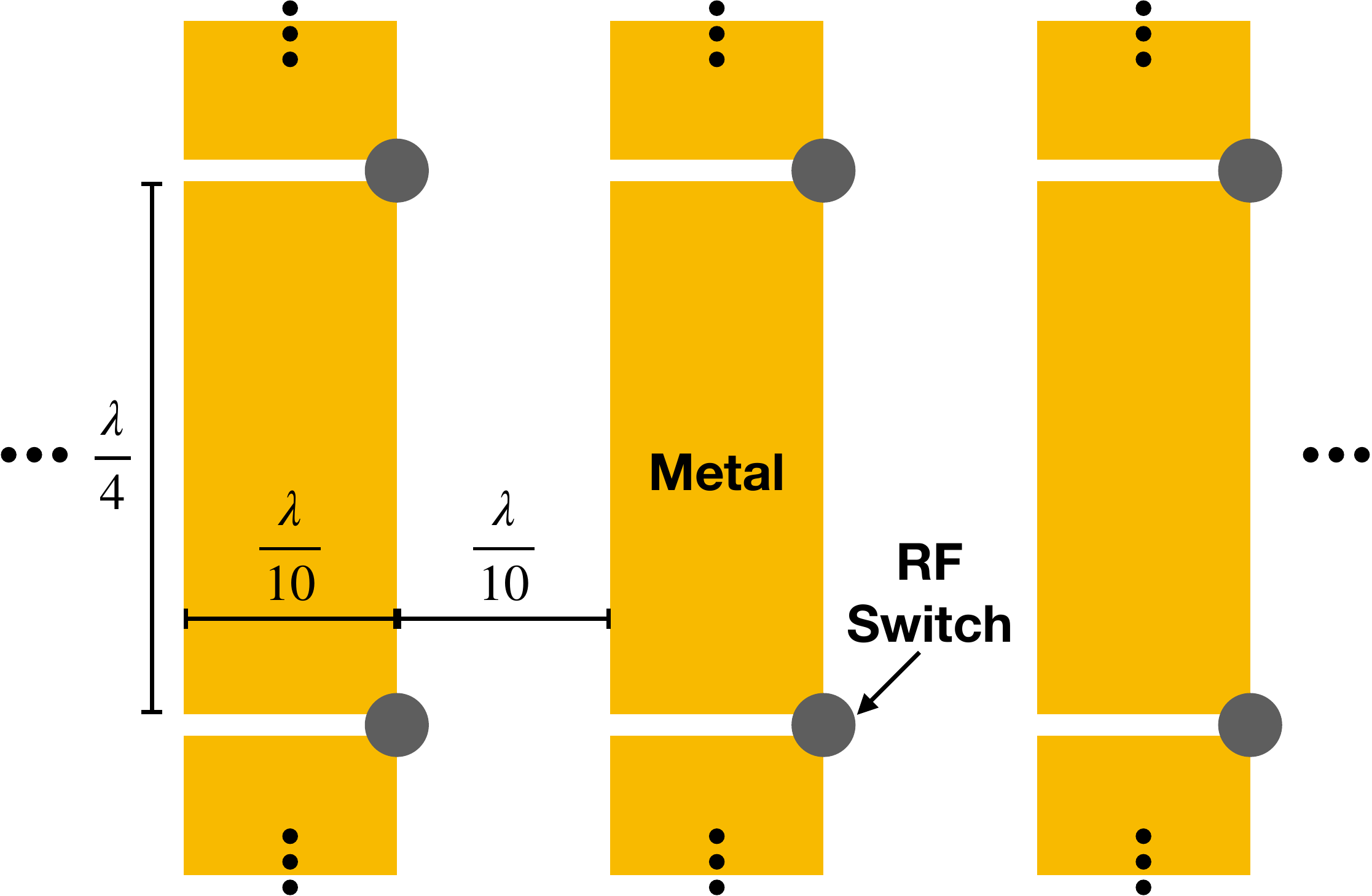}
\caption{Schematic of the design of our antenna array. This array of rectangles continues in both directions.}
\label{fig:phys-design}
\end{figure}
Our antenna array design consists of a panel of $\lambda/4 \times \lambda/10$ sized metal rectangles as shown in Figure~\ref{fig:phys-design}, where $\lambda$ denotes the wavelength. They are connected by RF-switches which determine whether or not the rectangles are connected (the switches are placed off-center for practical PCB-design reasons). This design works only for vertically polarized radiation. It can be generalized to all polarizations by having an identical copy perpendicular to this one.

There are two principles of operation. First, the rectangles by themselves are too small to strongly interact with radiation. However, when an RF switch is turned on, it joins two rectangles to form a $\lambda/2 \times \lambda/10$ rectangle. This forms a half-dipole antenna and interacts strongly with incident radiation. We made the strips wide to support a wider bandwidth of operation. The second principle is that, if a plate of metal has small holes in it, then radio behaves as if the holes weren't there. This is the same reason why microwave oven doors have small holes in them, and why airport radars can use a metal cage, and not a solid sheet, to form their rotating dish antenna. A commonly used rule-of-thumb states that the holes need to be smaller than $\lambda/10$. This motivates our choice of gaps between rectangles. When switches in adjacent columns are turned on, their rectangles should behave as a continuous sheet of metal, rather than individual columns, allowing us a large control over incident energy. Because  neighbors act in a simple way, we expect that the neighbors' state wouldn't change the phase of the currents induced in the rectangle, only the magnitude. Hence the linear model in equation~\ref{eqn:base-model} should be approximately correct.

The above reasoning is merely the conjecture that motivated our design. Simulating such a large array is very computationally intensive, therefore experiment is the most tractable option. We conduct two experiments to partially validate this theory.  \S\ref{s:eval:micro:linearity} demonstrates that the linear model is approximately correct, and~\S\ref{s:eval:micro:opacity} shows that the surface has a large control over the incident energy. Validating this design in an anechoic chamber would offer more insights, but we leave that for future work.

\va{Talk about doorknob example?}

%% file: eval.tex
\section{Evaluation}
\subsection{Experiment Setup}

Our antennas are fabricated on custom printed circuit boards, with 40 antennas per board. We mount 94 of these boards on metal frames and place it next to a wall in our lab. The boards are connected in series with a single SPI bus composed of shift registers, allowing our Arduino controller to set the state of each individual element. Setting one state takes $6\,ms$; generating 3760 pseudorandom bits and pushing them at 20Mbit/s though the serial bus consume most of the time. This is the primary bottleneck of our system. However this is a limitation of our implementation. One could imagine other architectures which are faster. For instance, rather than have one long serial bus, we could have multiple buses, each controlling a subset of the bits. Alternatively, each element could be similar to an RFID tag that sets its state to a pseudo-random number in response to a clock signal from the RFID-reader-like controller. If each element switched once every $10\,\mu s$ (the RF switch takes $< 1\,\mu s$), then we could do 10000 measurements every 100 ms, which is enough for typical indoor movements, since optimizing the surface requires $O(N)$ measurements ($N$ is the number of elements). Nevertheless, supporting motion is challenging with our architecture. For objects moving faster than walking speed, the easiest option would be to beamform to a large area, and give up on focusing to a small spot.

\paragraph{Measuring Phase Change} Some of our microbenchmarks measure the change in {\em phase} of the channel. To prevent this already weak signal from being further corrupted by drift in carrier frequency offset (CFO drift), it is important that the Arduino controller's time is synchronized to the radio's~\S\ref{s:algo:measurement}. To do so, we connect a wire from the Arduino to an RF-switch connected to the transmitter; The Arduino periodically switches off transmission, which the receiver detects to synchronize its clock to the Arduino's to within 30 $\mu$s. To give the phase-change of a state an absolute value, we always measure it relative to the channel where all elements are off, $h_Z$; If the phase of an assignment was $h$, we measure $h / h_Z$. Note, the synchronization is only for microbenchmarks. The main optimization algorithm relies purely on RSSI measurements, and does not require such fine synchronization.


\subsection{Microbenchmarks}
\label{s:eval:micro}

\subsubsection{Linearity}
\label{s:eval:micro:linearity}
Our optimization algorithm (\S\ref{s:algo:rssi-algo}) assumes that the linear model in Equation (\ref{eqn:base-model}) is correct, which states that the elements do not interact with each other. That is, the $h_i$ for one element does not depend on $b_i$ for any other element. Since the elements are flat and only radiate perpendicular to the surface, we would expect this to be true for elements that are separated by more than a few wavelengths. However nearby elements may interact, especially since we placed them close together in order 	to be able to control a large fraction of incident energy. We had an intuition that the non-linearity due to this interaction should be small (\S\ref{s:phys:design}). Now, we experimentally test this intuition.

\begin{figure}
\centering
\begin{tabular}{|c|c|}
\hline
{\bf Total prediction error} & {\bf Error due to  noise} \\
\hline
$5.4\%$ & $2.0\%$ \\
\hline
\end{tabular}
\caption{A linear model predicts the channel due to a state with $5.4\%$ accuracy. If the surface were perfectly linear, the error would have been $2.0\%$ due to noise. Hence the \name{} is approximately, but not fully, linear.}
\label{fig:linearity}
\end{figure}

We prepare several random ``test-triples'' of states of the form $(S_A, S_B, S_{AB})$. $S_A$ and $S_B$ are mutually exclusive. That is, no two elements are `on' in both of them, and any element that is `on' in either $S_A$ or $S_B$ is also `on' in $S_{AB}$. Treating the states as bit-strings and $\&$ and $|$ as bitwise operators, $S_A \& S_B = 0$ and $S_{AB} = S_A | S_B$. Let $(h_A, h_B, h_{AB})$ be the ground-truth channels for the triple, if our linear model in equation~\ref{eqn:base-model} is correct, $h_A / h_Z + h_B / h_Z  - 1= h_{AB} / h_Z$. Our setup can measure $h_X$ for any state $X$, and we test how well we can predict $h_{AB} / h_Z$ given $h_A/h_Z$ and $h_B/h_Z$. To get reliable measurements of each of these ratios, we measure each ratio 100 times.

Non-linearity, if any, will arise when neighboring elements that weren't simultaneously `on' in $S_A$ or $S_B$, are both turned `on' in $S_{AB}$. To ensure there are many such cases, we randomly assign each element to $S_A$ or $S_B$. Then we randomly choose the value for elements in $S_A$ and $S_B$, and compute $S_{AB}$ using a bit-wise or. This way, many new neighbors will interact in $S_{AB}$. As shown in Figure~\ref{fig:linearity}, we can predict the value of $S_{AB}$ with 5.4\% error. Hypothetically, if the surface were perfectly linear, we would have $2.0\%$ error due to noise. Hence, though the \name{} isn't perfectly linear, the non-linearities are small.

\subsubsection{Controllability and Bandwidth}
\label{s:eval:micro:opacity}
\begin{figure}
\includegraphics[width=\linewidth]{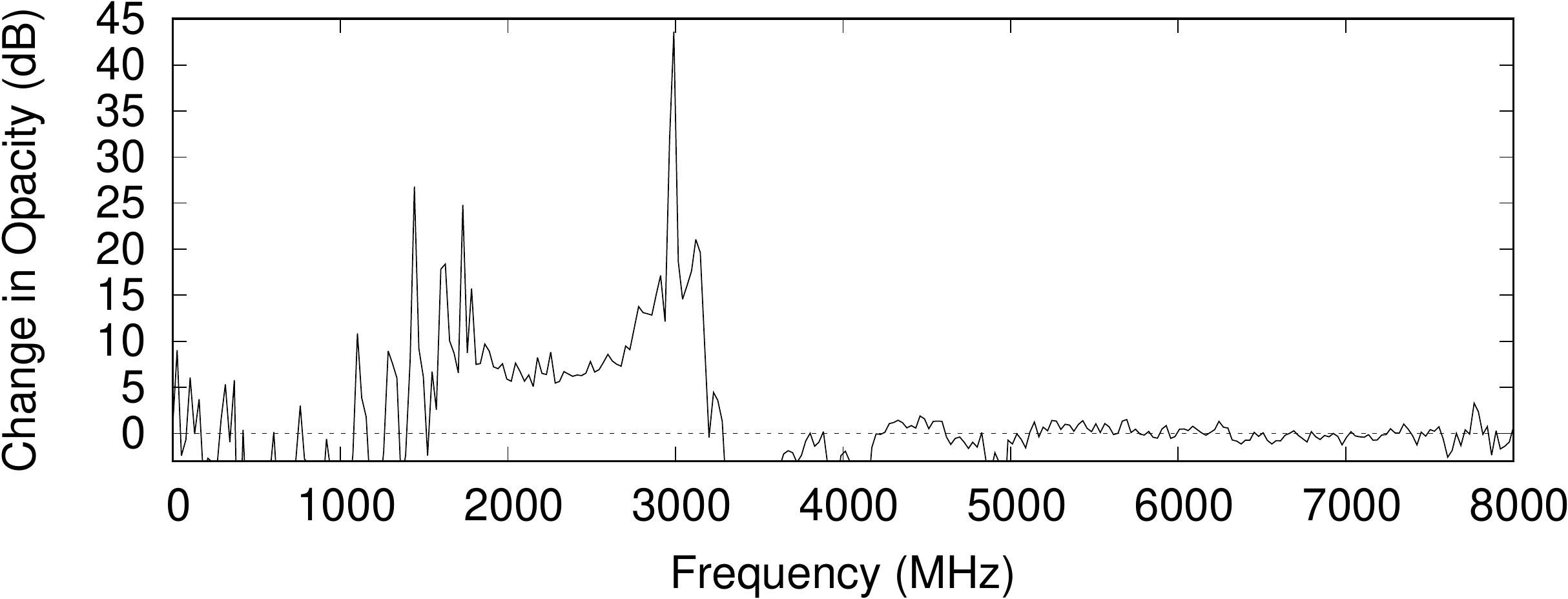}
\caption{The ability to control the surface's opacity to radio as a function of frequency.}
\label{fig:opacity}
\end{figure}

One design goal of our antenna array design was that it should be able to control a large fraction of incident energy. To test this, we kept the surface in between two wide-bandwidth directional (Vivaldi) antennas pointed at each other. Using a Vector Network Analyzer (VNA), we compare the signal strength between the antennas when all the elements are turned on and when they are all turned off. We expect that, when the elements are all on, the surface will be much more opaque to radiation, reflecting a large fraction of it. The ratio of signal strength in these configurations is shown in Figure~\ref{fig:opacity}. As shown, this ratio is consistently greater than 6 dB between 1600 and 3100 MHz. Hence, it can change its opacity by over 75\% over a large bandwidth. The peak is closer to 3000 MHz, where the change is well over 10 dB (90\% control). But all of our other results are in the 2450 MHz ISM band, in order to conform to FCC rules; we expect better results if we had operated closer to the peak. Frequency of operation can be tuned by scaling the sizes of the components. Since antenna design was not the focus of our current work, we leave this to future work.

The y-axis is cropped at $-3$ dB for clarity in showing our frequency range of interest. The change falls after 3000 MHz because our RF switch is only rated up to that level. At $< 1500$ MHz, the rectangles, even after joining, are too small to interact with radiation.

\subsubsection{Measurability}
\label{s:eval:micro:measure}
\begin{figure}
\includegraphics[width=\linewidth]{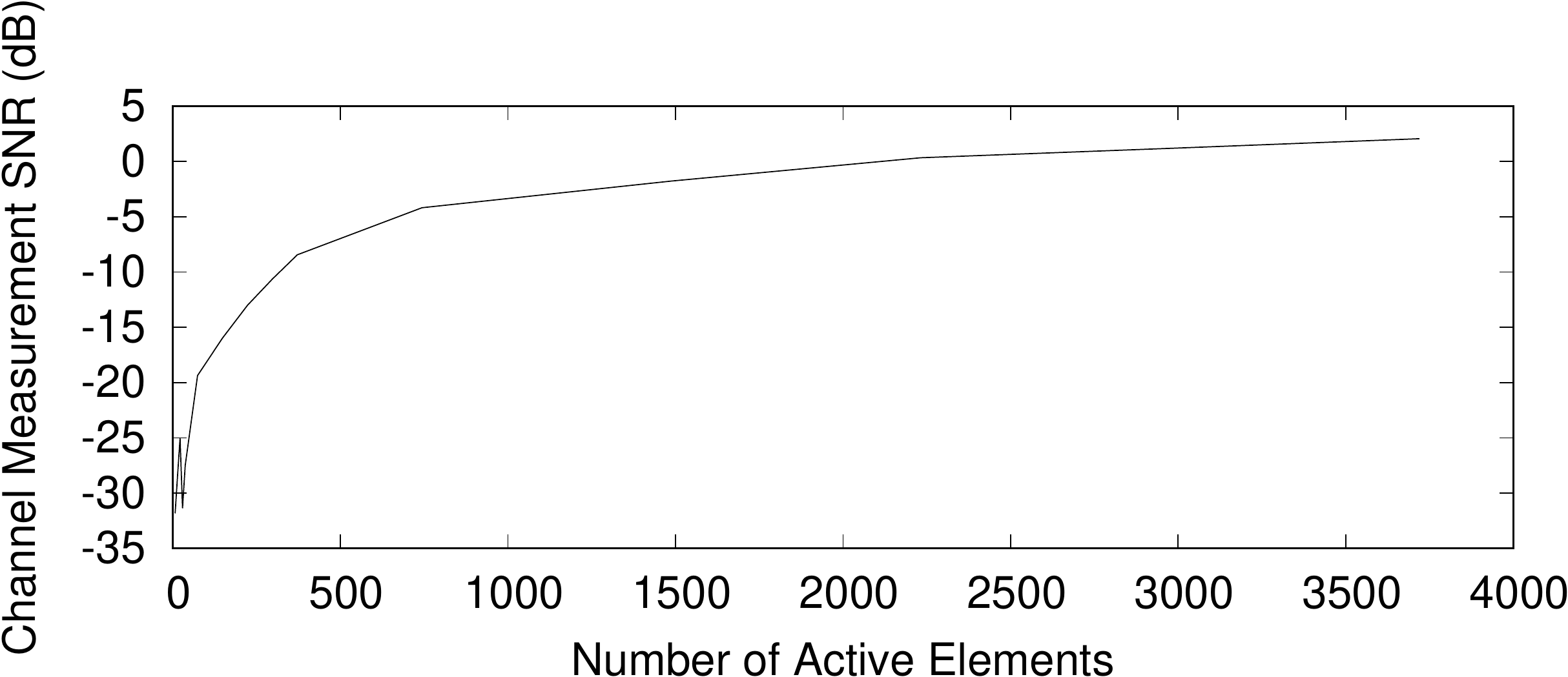}
\caption{Measurability of the effect random configurations of the \name{} surface on the channel, as a random subset of them are deactivated. It demonstrates why boosting~(\S\ref{s:algo:measurement}) is important for \name{} to function.}
\label{fig:measurability}
\end{figure}
To aid measurement of the effect of the \name{} surface on the channel, we vary all elements randomly and at once~(\S\ref{s:algo:measurement}). This gives us an $O(N)$ boost in our ability to measure the change in the channel. To experimentally study the effect of this boost, we compute the Signal to Noise Ratio (SNR) of the measurement as a function of the number of elements in the array. We artificially reduce the size of our array, by deactivating a random subset of it. 

We measure the RSSI-ratio for 100 different random configurations, repeating the measurement for each configuration 125 times. For each configuration, the RSSI-ratio is averaged over all measurements. The `Signal' in SNR is the variance in the (average) RSSI-ratio across all configurations, and the `Noise' is the average variance in the RSSI-ratio measurements within measurements of each individual configuration. We plot the SNR as a function of the number of active elements in Figure~\ref{fig:measurability}. In this experiment, the transmitter and receiver are on the same side of the \name{} surface, separated by about 1 meter; when the \name{} surface is optimized for this pair, it achieves a $12\times$ gain in RSSI.

We can see that SNR is much higher with a greater number of elements. Our estimate of SNR is only reliable above $-20\,dB$, hence the effect of varying just one element is well below our ability to detect. This is why boosting the signal by varying all elements at once, is critical.

Nevertheless, even at its highest point, measurability is still at $2\,dB$, which is why it is important that our algorithm be able to use noisy measurements. Note, this is the impact of random configurations on the channel. When optimized to eliminate destructive interference, the effect is amplified by $O(N)$, which is why the optimized state still produces significant gains in signal strength, though a random state doesn't have much impact.

\subsection{Signal Strength Optimization}
\label{s:eval:main}

\begin{figure}
\includegraphics[width=\linewidth]{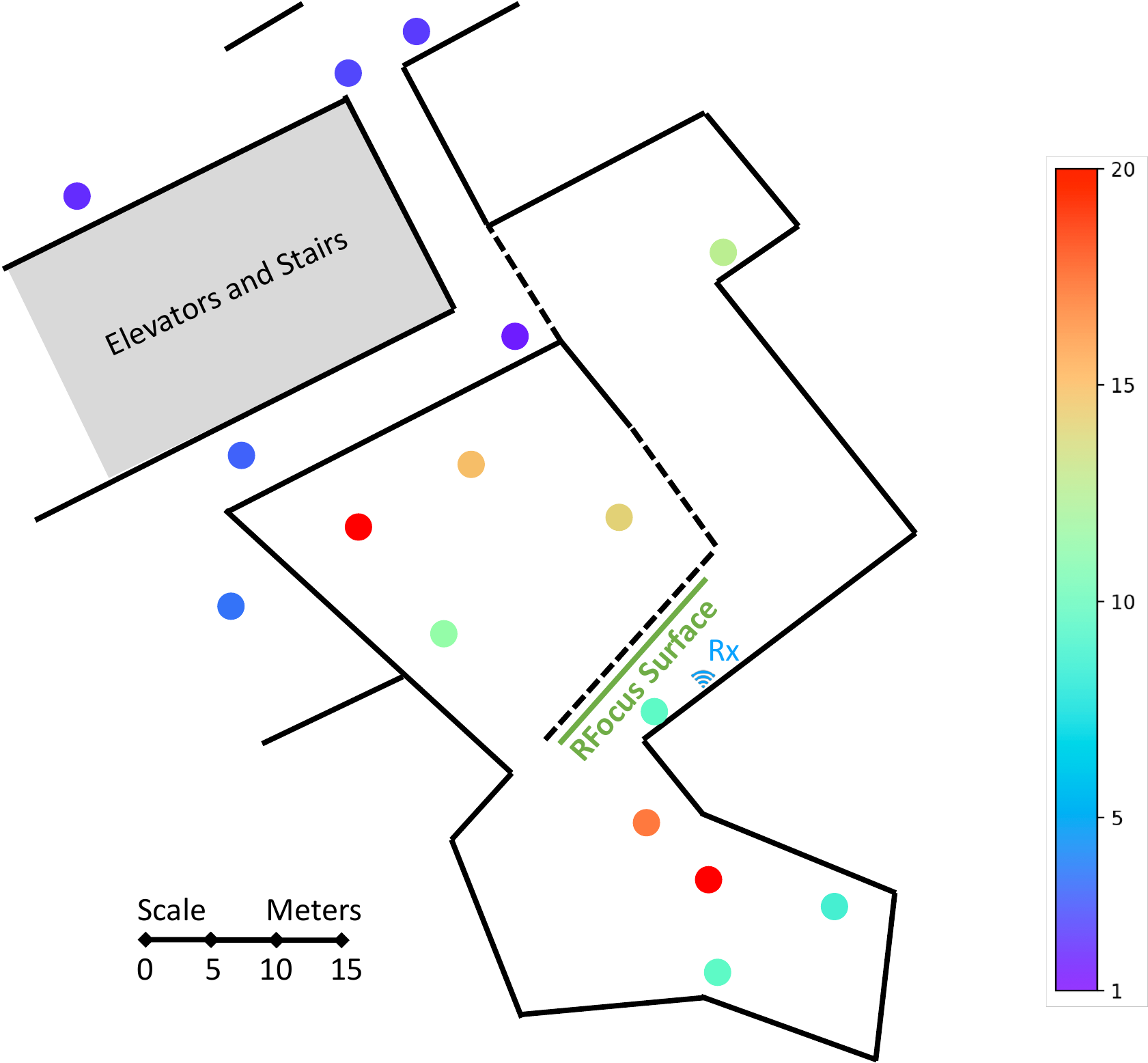}
\caption{Map of signal strength improvement (ratio). For the red point on the bottom left, the improvement is $64\times$, which is clipped on the map. The receiver and \name{} surface are in fixed positions, as shown. Signal strength improvement is plotted for various points on the map. The dotted lines indicate thinner, glass walls. A CDF is shown in~\ref{fig:main-eval:cdf}.}
\label{fig:main-eval:map}
\end{figure}

\begin{figure*}
\centering
{\subfigure[CDF of signal strength improvement. Median improvement is $10.5\times$]{
\includegraphics[width=.45\textwidth]{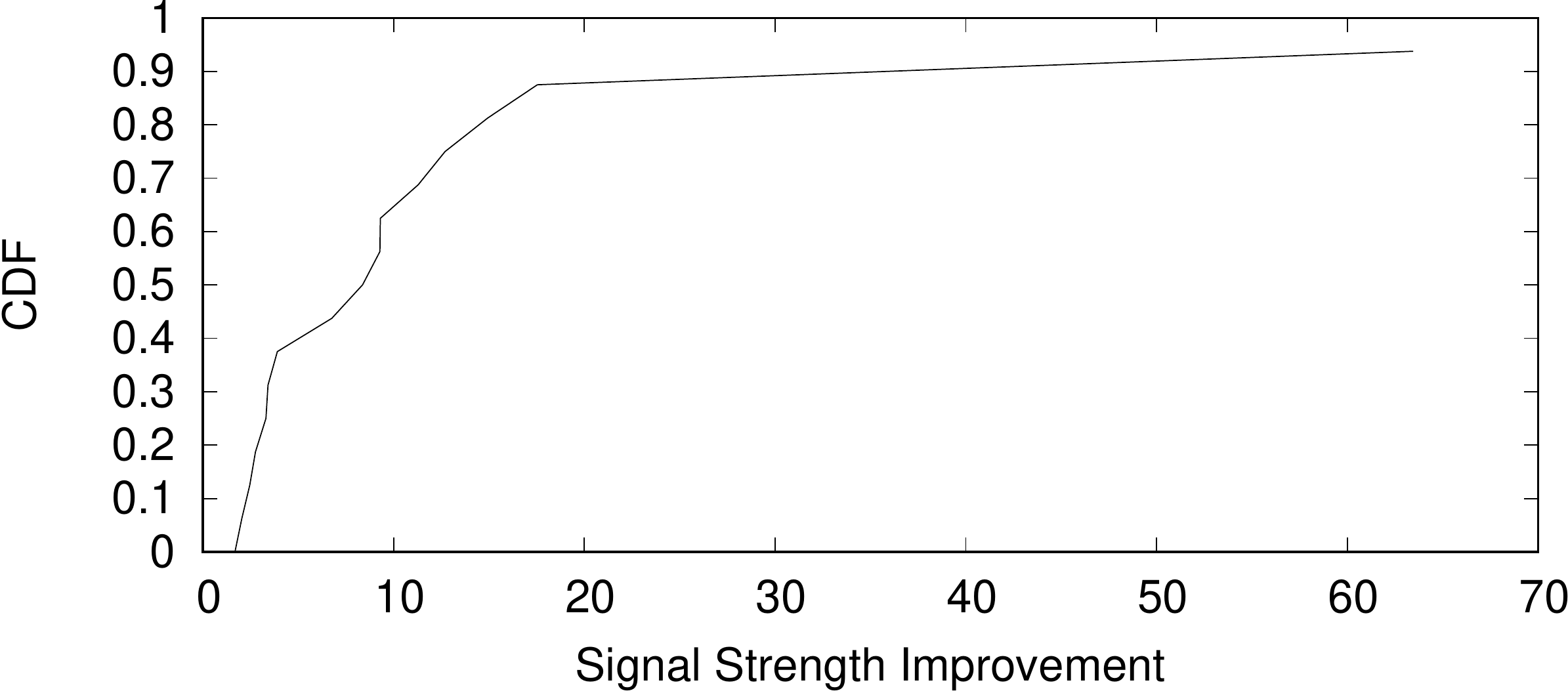}
\label{fig:main-eval:cdf}
}}
{\subfigure[CDF of channel capacity improvement. Median improvement is $2.1\times$]{
\includegraphics[width=.45\textwidth]{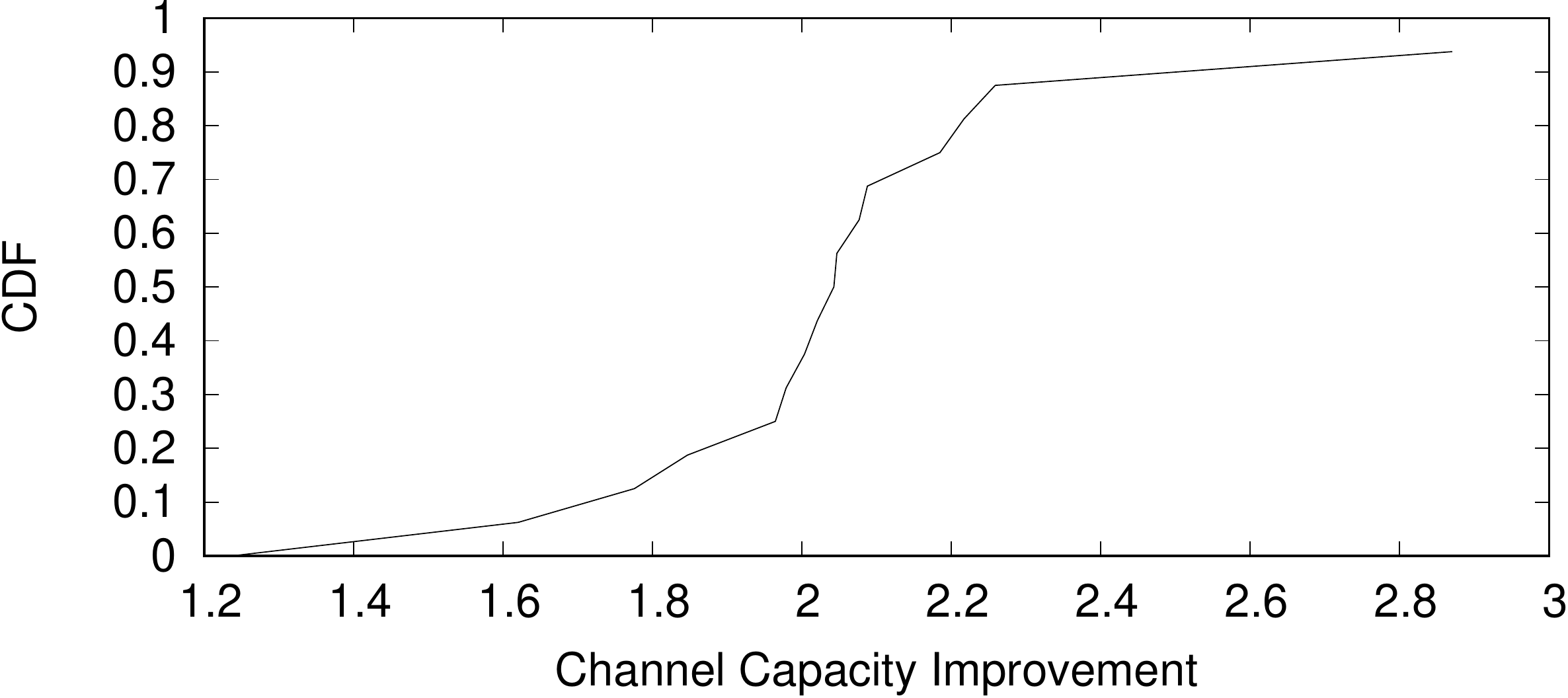}}
\label{fig:main-eval:cdf-tpt}}
\caption{Improvement in the signal strength and channel capacity. Notice that though the locations in the top left on the map~(Figure~\ref{fig:main-eval:map}) do not achieve much signal strength improvement, their baseline SNR is also low, which leads to a $\approx 2\times$ channel capacity improvement.}
\label{fig:main-eval}
\end{figure*}

For our main evaluation, we placed the receiver and the \name{} surface at a constant location as marked in Figure~\ref{fig:main-eval:map}. Then we placed the transmitter various positions in an indoor environment (our lab), and ran the optimization algorithm to maximize signal strength at the receiver. We measure the ratio of the improved signal strength to the signal strength when the all elements of the \name{} are `off'. We plot these in a map in Figure~\ref{fig:main-eval:map}. The corresponding CDFs are shown in Figure~\ref{fig:main-eval}. \name{} consistently achieves a $\approx 10\times$ improvement in signal strength for all points not occluded by a major wall/elevators, as marked by solid lines. For occluded points, it achieves a $2-4\times$ improvement. In all cases, it achieves $\approx 2\times$ improvement in throughput.

\name{} is able to achieve these benefits because its large area allows it to precisely focus energy from the transmitter to the receiver. This is particularly helpful when the transmitters are power constrained, since even a `whisper' will be `heard' clearly at the receiver. Yet, interference is minimized since the transmit power is low. This could enable a new regime of low-power, high throughput IoT sensor devices. Whenever the receiver wants to probe data from a sensor, it can ask the controller to tailor the surface for that particular endpoint. This takes $\approx\,1\,ms$. Then it can initiate communication with the sensor, which can transmit at low power. Since sensors do not tend to move, the same trained configuration can be used for extended periods of time. Most of our experiments we conducted during regular work-hours, with the usual indoor motion. The pictured area has $\approx 15$ workspaces. Some of our trained configurations worked across multiple days, as long as the endpoints didn't move.

Note that \name{} functions as both a mirror, when the receiver on the same side of the surface as the transmitter, and as a lens, when it is not.

\subsubsection{Improvement Across Frequencies}
\begin{figure}
\includegraphics[width=\linewidth]{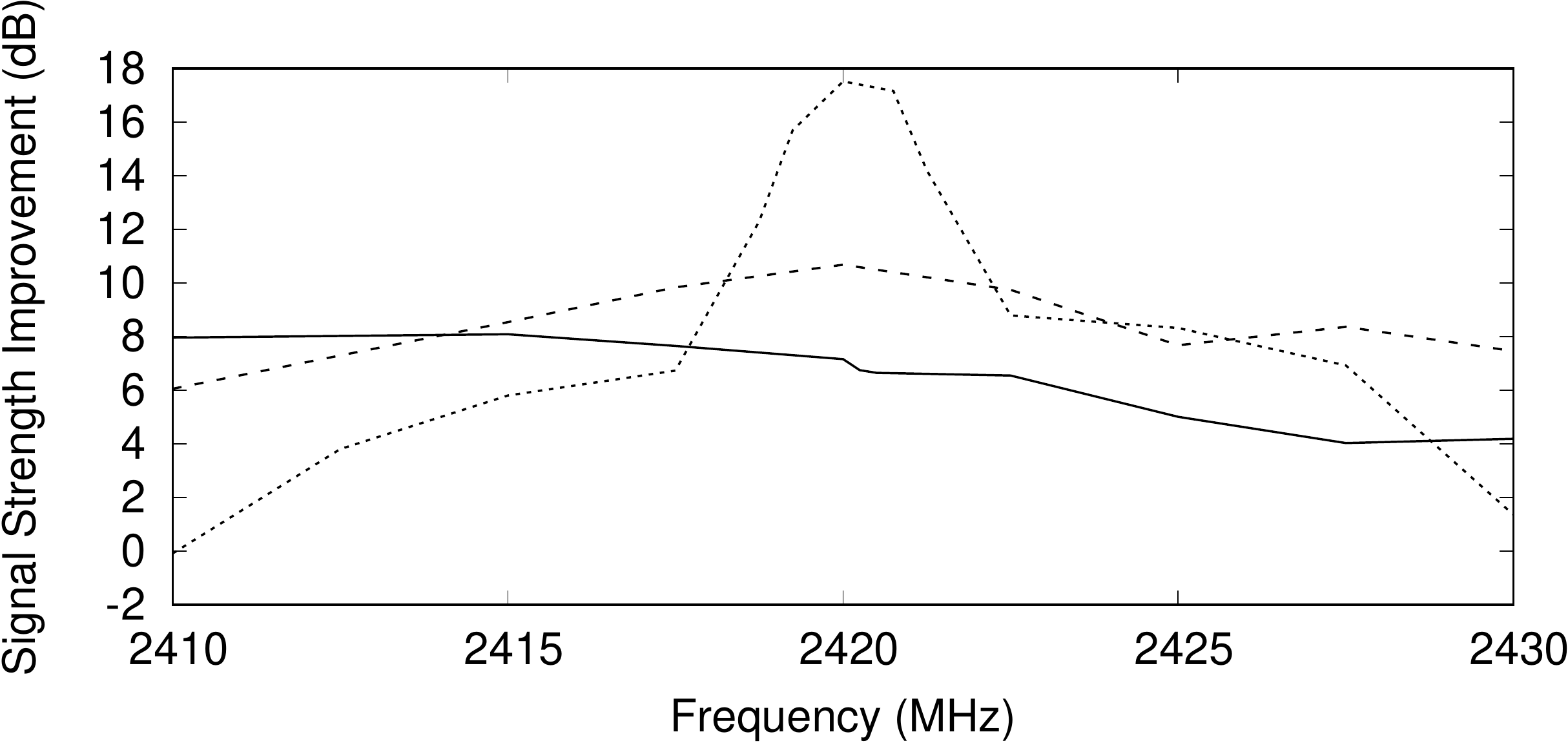}
\caption{Signal strength improvement as a function of frequency when RFocus was optimized for 2420 MHz. \name generalizes across nearby frequencies, even though the optimization algorithm only sought to optimize at one frequency.}
\end{figure}

Our optimization algorithm only seeks to improve the signal strength at a single frequency. Nevertheless, we find that it also provides improvement for nearby frequencies. We plot the improvement as a function of frequency for three different links, where the target optimization frequency was 2420 MHz, but we obtain benefits in a 20 MHz neighborhood.

This happens because \name{} benefits from spatial diversity. Since it has many paths between the same pair of points, each of them experiences fading separately, and it is unlikely that all of them fade simultaneously. The difference in wavelength between the two frequencies is only 1/120 of a wavelength. Hence, the optimal configuration for nearby wavelengths will be similar. Only when the spacing between a pair of elements is large enough for this difference to accumulate to $\approx \frac{1}{2}\lambda$, will the optimal configuration change appreciable. The natural frequency width of a \name{} optimized state is sufficient for WiFi channels upto 20MHz. We leave optimizing the \name{} for a wider range of frequencies for future work. 

\subsubsection{Quadratic Growth}
\begin{figure}
\includegraphics[width=\linewidth]{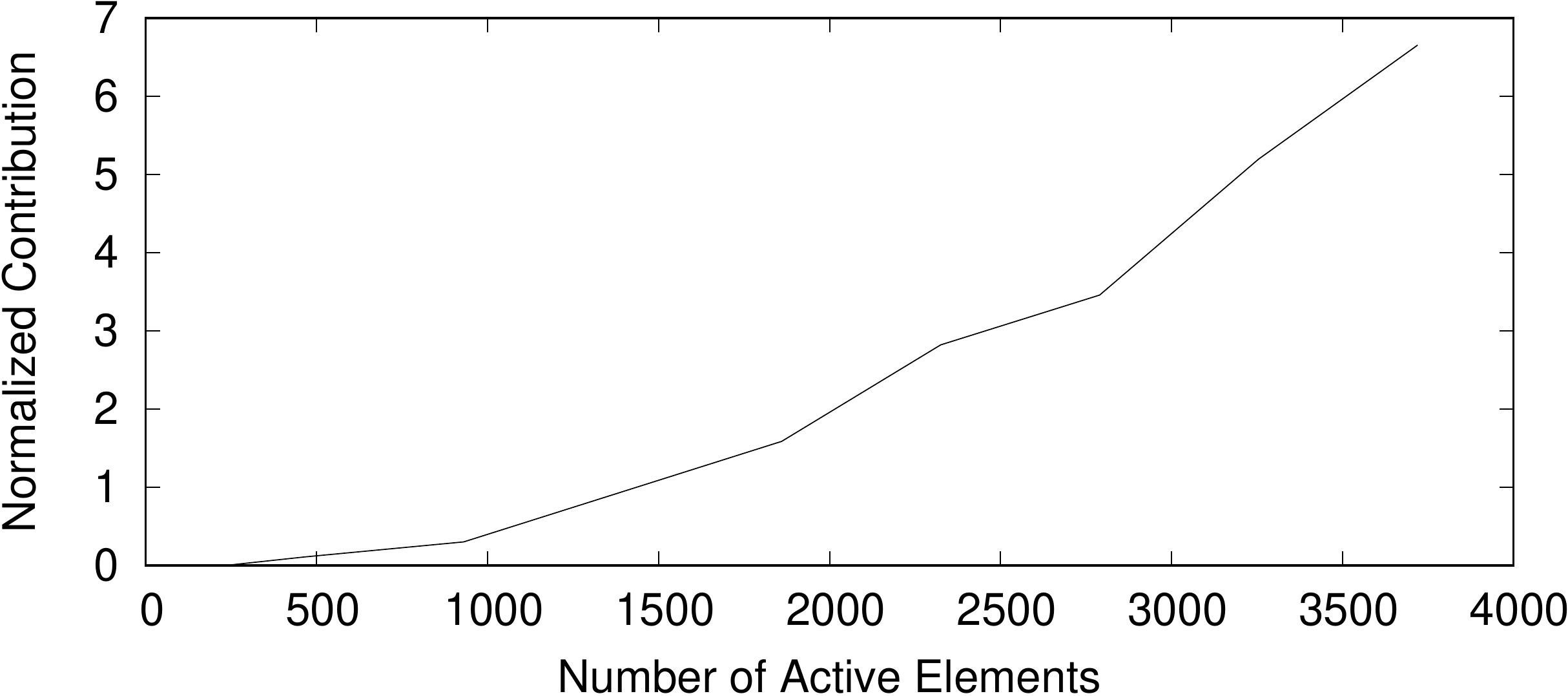}
\caption{The benefit of the \name{} surface grows quadratically with the number of elements. The y-axis shows the contribution to the signal strength of paths going via the \name{} surface.}
\label{fig:quadratic}
\end{figure}
Our model of the system suggests that the signal strength increases quadratically with the number of elements~(\S\ref{s:size-tradeoff}). This is because each element contributes linearly to the channel amplitude, and the signal strength is the sqare of the amplitude. To experimentally verify this phenomenon, we first trained the \name{} surface's configuration for a pair of endpoints. Then we artificially disabled a random subset of elements, and computed the improvement in signal strength, {\em due to the \name{} surface}. That is, we plot the signal strength after discounting the effect of $h_Z$. Figure~\ref{fig:quadratic} demonstrates this quadratic increase for one pair of endpoints, as a function fo the number of active elements.

\subsubsection{Optimization Speed}
\label{s:eval:opt-speed}

\begin{figure}
\includegraphics[width=\linewidth]{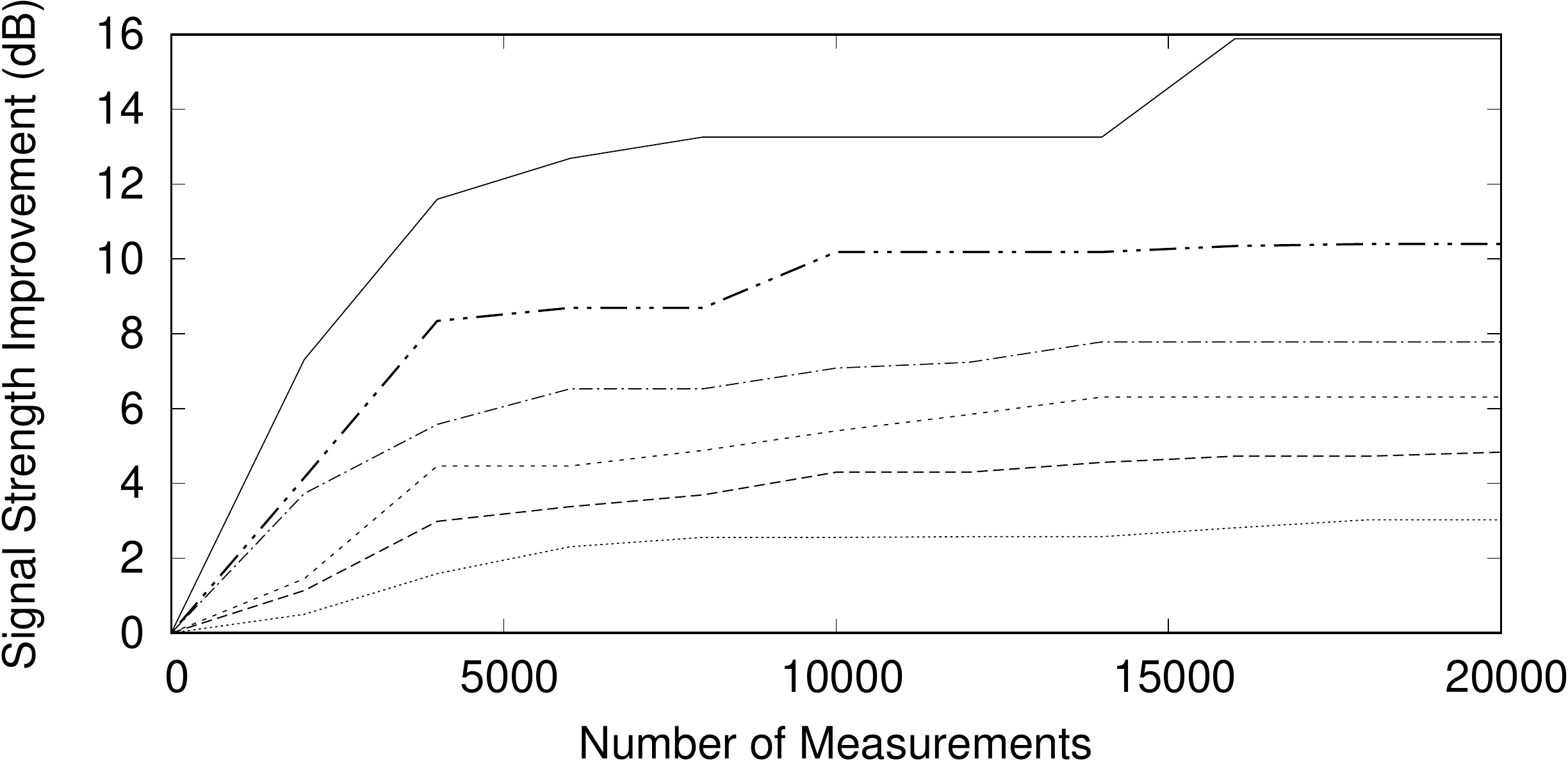}
\caption{The increase in signal strength due to the best learned configuration, as a function of the number of measurements for 6 pairs of endpoints. The increase is computed relative to the channel when the entire \name{} surface is turned off. The pairs span the entire range of performance, from 3 to 15 $dB$. Measurements occur in batches of 2000~(\S\ref{s:algo:rssi-algo}) As we can see, most of the gain comes early on, before the entire surface is fully optimized.}
\label{fig:opt-speed}
\end{figure}

To understand the rate at which the optimization progresses, we plot the signal strength improvement as a function of the number of measurements in Figure~\ref{fig:opt-speed}. As shown, most of the improvement occurs with 4000 measurements. Note, \name{} has 3720 elements, and we'd expect to need at-least 3720 measurements before the problem can be well specified, even ignoring noise. However, the \name{} optimization algorithm can get benefits earlier, since it fixes the state of a reflector, as soon as it is 95\% confident about it. At any point in time, it has a hypothesis state that it believes is optimal.

%% file: concl.tex
\section{Conclusion}

This paper presented RFocus, a system that moves beamforming functions from the radio transmitter to the environment. RFocus includes a two-dimensional surface with a rectangular array of simple elements, each of which functions as an RF switch. Each element either lets the signal through or reflects it. The state of the elements is set by a software controller to maximize the signal strength at a receiver, using a majority-voting-based optimization algorithm. The RFocus surface can be manufactured as an inexpensive thin wallpaper, requiring no wiring. Our prototype implementation improves the median signal strength by $10.5\times$, and the median channel capacity by $2.1\times$. 


\paragraph{Human Safety.} Because the \name{} surface doesn't emit any energy of its own, it does not increase the total radiation. It can focus the energy to an area the size of a wavelength, which is as risky as being near the transmitter. \name{} appreciably increases signal strength only near the intended receiver, and not at other locations. Any device's ability to focus energy to an area smaller than a wavelength drops exponentially with distance from it.

\paragraph{Ethics Statement.} This work complies with all applicable ethical standards of our home institution, including (but not limited to) privacy policies and policies on experiments involving humans. No human subjects were involved in this research.

%% file: acks.tex
\section{Acknowledgements}
This project was made possible by many interesting discussions with: Dinesh Bharadia, Peter Iannucci, Zach Kabelac, Dina Katabi, Colin Marcus, Vikram Nathan, Hariharan Rahul, Deepak Vasisht and Guo Zhang. We would also like to thank members of the NetMIT lab for letting us borrow their radio equipment.

%% file: main.bbl
\begin{thebibliography}{10}

\bibitem{3d-print-reflect-1}
J.~Chan, C.~Zheng, and X.~Zhou.
\newblock {WiPrint: 3D Printing Your Wireless Coverage}.
\newblock In {\em {HotWireless}}, 2015.

\bibitem{reconf-ant-review}
J.~Costantine, Y.~Tawk, S.~E. Barbin, and C.~G. Christodoulou.
\newblock Reconfigurable antennas: Design and applications.
\newblock {\em Proceedings of the IEEE}, 103(3):424--437, 2015.

\bibitem{ee-3-theory}
M.~Di~Renzo and J.~Song.
\newblock Reflection probability in wireless networks with metasurface-coated
  environmental objects: An approach based on random spatial processes.
\newblock {\em arXiv preprint arXiv:1901.01046}, 2019.

\bibitem{refarr-ant-review}
J.~Huang and J.~A. Encinar.
\newblock {\em Reflectarray antennas}, volume~30.
\newblock John Wiley \& Sons, 2007.

\bibitem{laia-nsdi}
Z.~Li, Y.~Xie, L.~Shangguan, R.~I. Zelaya, J.~Gummeson, W.~Hu, and K.~Jamieson.
\newblock Towards programming the radio environment with large arrays of
  inexpensive antennas.
\newblock In {\em NSDI}, 2019.

\bibitem{hypersurface-3}
C.~Liaskos, S.~Nie, A.~Tsioliaridou, A.~Pitsillides, S.~Ioannidis, and
  I.~Akyildiz.
\newblock A new wireless communication paradigm through software-controlled
  metasurfaces.
\newblock {\em IEEE Communications Magazine}, 56(9):162--169, 2018.

\bibitem{hypersurface}
C.~Liaskos, S.~Nie, A.~Tsioliaridou, A.~Pitsillides, S.~Ioannidis, and
  I.~Akyildiz.
\newblock A novel communication paradigm for high capacity and security via
  programmable indoor wireless environments in next generation wireless
  systems.
\newblock {\em Ad Hoc Networks}, 87:1--16, 2019.

\bibitem{hypersurface-2}
C.~Liaskos, A.~Tsioliaridou, S.~Nie, A.~Pitsillides, S.~Ioannidis, and
  I.~Akyildiz.
\newblock Modeling, simulating and configuring programmable wireless
  environments for multi-user multi-objective networking.
\newblock {\em arXiv preprint arXiv:1812.11429}, 2018.

\bibitem{angular-spectrum-representation}
L.~Novotny.
\newblock {\em {Lecture Notes on Electromagnetic Fields And Waves}}.
\newblock 2013.

\bibitem{laia-phase-shifter}
{MACOM Maps-010144 four-bits phase shifter}.
\newblock \url{http://cdn.macom.com/datasheets/maps-010144.pdf.},.

\bibitem{rf-switch}
{JSW2-33DR-75+ SPDT RF Switch}.
\newblock
  \url{https://www.minicircuits.com/WebStore/dashboard.html?model=JSW2-33DR-75}.

\bibitem{ee-2}
L.~Subrt and P.~Pechac.
\newblock Intelligent walls as autonomous parts of smart indoor environments.
\newblock {\em IET Communications}, 6(8):1004--1010, 2012.

\bibitem{ee-1-with-exp}
X.~Tan, Z.~Sun, J.~M. Jornet, and D.~Pados.
\newblock Increasing indoor spectrum sharing capacity using smart
  reflect-array.
\newblock In {\em 2016 IEEE International Conference on Communications (ICC)},
  pages 1--6. IEEE, 2016.

\bibitem{laia-hotnets}
A.~Welkie, L.~Shangguan, J.~Gummeson, W.~Hu, and K.~Jamieson.
\newblock Programmable radio environments for smart spaces.
\newblock In {\em HotNets}, 2017.

\bibitem{3d-print-reflect-2}
X.~Xiong, J.~Chan, E.~Yu, N.~Kumari, A.~Sani, C.~Zheng, and X.~Zhou.
\newblock Customizing indoor wireless coverage via 3d-fabricated reflectors.
\newblock In {\em 4th ACM International Conference on Systems for
  Energy-Efficient Built Environments}, 2017.

\end{thebibliography}
